\documentclass[journal]{IEEEtran}

\usepackage{amsmath,amssymb,amsthm,amsfonts,mathtools} 
\usepackage{dsfont,eucal,bbm,bm,nicefrac} 
\usepackage{graphicx,float,subcaption,booktabs} 
	\graphicspath{{./figures/}}
\usepackage{algorithm,algorithmic} 
\usepackage{hyperref} 
\usepackage{tikz,xcolor,soul} 
\usepackage{cite}
\usepackage{textcomp}
\usepackage{hyperref}
\hypersetup{hidelinks}
\usepackage{balance}

\newtheorem{theorem}{Theorem}
\newtheorem{corollary}{Corollary}[theorem]

\newtheorem{assumption}{Assumption}
\newtheorem{remark}{Remark}

\newtheorem{definition}{Definition}

\newcommand{\T}{\ensuremath{\mathrm{T}}} 
\newcommand{\dd}{\ensuremath{\mathrm{d}}} 

\newcommand{\pbra}[1]{\ensuremath{\left( #1\right)}}
\newcommand{\sbra}[1]{\ensuremath{\left[ #1\right]}}
\newcommand{\cbra}[1]{\ensuremath{\left\{ #1\right\}}}
\newcommand{\abra}[1]{\ensuremath{\left< #1\right>}}

\newcommand{\pder}[2]{\ensuremath{\frac{\partial #1}{\partial #2}}}
\newcommand{\E}[1]{\ensuremath{\mathbb{E}\left[ #1\right]}}

\newcommand{\Ep}[2]{\ensuremath{\mathbb{E}_{#1}\left[ #2\right]}}

\DeclareMathOperator*{\argmin}{arg\,min}

\DeclareMathOperator*{\minimize}{minimize}
\DeclareMathOperator*{\maximize}{max.}

\begin{document}

\title{
Real-time Hybrid System Identification \\
with Online Deterministic Annealing}

\author{Christos N. Mavridis$^*$, 
and Karl Henrik Johansson$^*$  
%
\thanks{$^*$%
Division of Decision and Control Systems, 
School of Electrical Engineering and Computer Science,
KTH Royal Institute of Technology, Stockholm.
{\tt\small emails:\{mavridis,kallej\}@kth.se}.}%
%
%
\thanks{%
Research partially supported 
by the Swedish Foundation for Strategic Research (SSF) grant IPD23-0019.%
}%
}

\maketitle
 \thispagestyle{empty}
\pagestyle{empty}

\begin{abstract}
We introduce a real-time identification method for discrete-time state-dependent switching systems in both the input--output and state-space domains. In particular, we design a system of adaptive algorithms running in two timescales; a stochastic approximation algorithm implements an online deterministic annealing scheme at a slow timescale and estimates the mode-switching signal, and a recursive identification algorithm runs at a faster timescale and updates the parameters of the local models based on the estimate of the switching signal. We first focus on piece-wise affine systems and discuss identifiability conditions and convergence properties based on the theory of two-timescale stochastic approximation. In contrast to standard identification algorithms for switched systems, the proposed approach gradually estimates the number of modes and is appropriate for real-time system identification using sequential data acquisition. The progressive nature of the algorithm improves computational efficiency and provides real-time control over the performance-complexity trade-off. Finally, we address specific challenges that arise in the application of the proposed methodology in identification of more general switching systems. Simulation results validate the efficacy of the proposed methodology.
\end{abstract}

\begin{IEEEkeywords}
Switched System Identification, Piecewise Affine System Identification, Online Deterministic Annealing.
\end{IEEEkeywords}


\vspace{-1em}
\section{Introduction}
\label{Sec:Introduction}

Switched systems, described by interacting continuous and discrete dynamics, are
a powerful modeling tool in the analysis of systems 
where logic and continuous processes are interlaced, as in most complex cyber-physical systems.
In addition to being able to describe switching dynamics, switched systems can be used as a tool to 
approximate highly non-linear dynamics by a collection of simpler models, and boost 
model explainability and robustness, by decomposing the behavior of a complex system into 
sub-systems where first principles and domain knowledge can be used for precise model tuning 
\cite{garulli2012survey,liberzon2003switching}.
As a result, switched systems have attracted significant attention in the control community.

However, first principles modelling is often too complicated and sub-optimal, and a switched
model needs to be identified on the basis of observations.
The majority of the work in this area is based on 
piece-wise affine (PWA) systems, a class of state-dependent switched systems
with important applications in identification, verification, and control synthesis 
of switched and nonlinear systems 
\cite{liberzon2003switching,paoletti2007identification,moradvandi2023models,bemporad2000observability}. 
PWA systems are a collection of affine dynamical systems, 
indexed by a discrete-valued switching variable (mode) that depends 
on a partitioning of the state-input domain
into a finite number of polyhedral regions \cite{liberzon2003switching,paoletti2007identification}. 
The input--output representation of PWA systems is the class of 
piece-wise affine auto-regressive exogenous (PWARX) systems with 
the switching signal depending on a partitioning of the domain of a vector 
containing the recent history of input--output pairs.
As the problem of identifying a PWA system can be challenging
\cite{vidal2002observability,petreczky2013realization},
%
most existing approaches focus on offline identification methods
\cite{roll2004identification,ferrari2003clustering}.

\vspace{-1em}
\subsection{Contribution and Outline}

%
In this work, we propose a two-timescale stochastic optimization approach for 
real-time state-dependent switched system identification in both input--output and state-space representations.
We first focus on the well-studied case of PWA and PWARX systems. In Section \ref{Sec:SwitchedPreliminaries} we present the 
realization and identifiability conditions for PWA systems, and in Theorem \ref{thm:identification}
of Section \ref{sSec:pwa-identifiability} 
we provide the identifiability conditions for state space PWA systems in the form of a persistence of excitation (PE) criterion.
In Section \ref{Sec:formulation}, we formulate the state-dependent switching system identification problem as a combined identification and prototype-based learning problem and in Sections \ref{Sec:ODA} and \ref{Sec:PWAID} we develop a two-timescale stochastic approximation algorithm to solve it in real-time.

In particular, in Section \ref{Sec:ODA} 
we build upon the online deterministic annealing approach \cite{mavridis2023annealing} to construct a stochastic approximation algorithm that estimates the mode-switching signal, as well as the number of modes, through a bifurcation phenomenon, studied in Section \ref{sSec:bifurcation}.
In Section \ref{Sec:PWAID} a second stochastic approximation algorithm based on standard 
adaptive filtering, running at a faster timescale,
is developed to update the parameters of the local models based on the estimate of the 
switching signal. 
The convergence properties of this system of recursive algorithms are studied in Theorem \ref{thm:two_timescales} of Section \ref{sSec:combined},
and the applicability of the proposed approach in more general state-dependent switching systems is discussed in Section \ref{Sec:SwitchedID}.
%
Finally, in Section \ref{Sec:Results}, simulation results validate the efficacy of the proposed approach in PWA systems.

\vspace{-1em}
\subsection{Related Work}

Most existing switched system identification methods can be categorized by 
the problem formulation used as
optimization-based \cite{bianchi2020alternating,roll2004identification,bako2011identification}, likelihood-based 
\cite{bemporad2018fitting,piga2020estimation,piga2020rao},
algebraic \cite{vidal2003algebraic,vidal2008recursive}, 
or clustering-based \cite{wang2020pwa,bako2019piecewise,ferrari2003clustering,gegundez2008identification,nakada2005identification}, and by 
the method used as offline \cite{bianchi2020alternating,li2021rapid,mejari2020recursive,ferrari2003clustering,juloski2005bayesian,bako2011recursive,vidal2008recursive} 
or online \cite{du2021novel}. 
%
For an extensive review of existing work the readers are referred to 
\cite{moradvandi2023models,lauer2019hybrid,garulli2012survey,paoletti2007identification,liberzon2003switching} and the references therein.

Algebraic methods are mainly based on transforming a Switched AutoRegressive eXogenous (SARX) model to a ``lifted'' ARX model 
that does not depend on the switching sequence \cite{vidal2003algebraic,vidal2008recursive}. 
Optimization-based methods rely on solving a large mixed-integer program, which is an NP hard problem that remains tractable only for simple models and small data sets \cite{bako2011identification,roll2004identification,paoletti2022iterative}.
Therefore, many works focus on relaxation techniques over the same problem \cite{tang2021expectation,hartmann2015identification,bako2011recursive}, 
that include convexification and expectation-maximization approaches.
Finally, clustering-based methods are optimization-based methods 
that make use of unsupervised learning to estimate the partition of the
domain that is needed for the switching signal 
\cite{wang2020pwa,ferrari2003clustering,gegundez2008identification,nakada2005identification,baptista2011split,ivanescu2011employing,bianchi2020alternating,yu2023randomized}.
%

Most hybrid identification approaches are offline methods
that first classify each observation 
and estimate the local model parameters (either simultaneously or iteratively), 
and then reconstruct the partition of the switching signal 
\cite{tang2021expectation,bianchi2020alternating,wang2020pwa,ferrari2003clustering,gegundez2008identification}. 
%
%
%
In our recent work, we have proposed the use of the online deterministic annealing approach as a
clustering method to estimate the partition of the switching signal in real-time 
\cite{mavridis2023identification,mavridis2024state}.
%
%
In this work, we extend these methods to
provide a complete study of a real-time prototype-based learning algorithm that (i) provides an inherent mechanism to adaptively estimate a set of modes with minimal cardinality, (ii) constitutes a unified switched system identification method for both input--output and state-space representations, and (iii) investigates extensions to more general switching systems.

\vspace{-1em}
\subsection{Notation}

The sets $\mathbb R$ and $\mathbb Z$ represent the sets of real and integer numbers, respectively, while
$\mathbb Z_+$ represents the set of non-negative integers. 
For a real matrix $A\in\mathbb{R}^{n\times m}$, 
$A^\T\in\mathbb{R}^{m\times n}$ denotes its transpose and $\text{vec}(A)\in\mathbb R^{mn}$ the vectorization of $A$.
The $n\times n$ identity matrix is denoted $I_{n}$.
$A \succeq 0$ is a positive semi-definite matrix, and the condition $A\succeq B$ is understood as $A-B\succeq 0$.
Unless otherwise specified, random variables $\mathcal{X}:\Omega\rightarrow \mathbb R^d$ are defined in 
a probability space $(\Omega,\mathbb F,\mathbb P)$.
The probability of an event is denoted 
$\mathbb P\sbra{\mathcal X\in S} \coloneq \mathbb P\sbra{\omega\in\Omega: \mathcal X(\omega)\in S}$, 
and the expectation operator $\E{\mathcal X} = \int_\Omega \mathcal X \textrm{d}\mathbb P$.
In case of multiple random variables $(\mathcal X,\mathcal Y)$ and a deterministic function $f$, 
the expectation operator $\E{f(\mathcal X,\mathcal Y)}$ 
is understood with respect to the joint probability measure, while 
$\E{\mathcal X|\mathcal Y} \coloneq \E{\mathcal X|\sigma(\mathcal Y)}$ denotes 
the expectation of $\mathcal X$ conditioned to the $\sigma$-field of $\mathcal Y$.
Stochastic processes $\cbra{\mathcal X(k)}_k$, $k\in\mathbb Z_+$, are defined in the filtered probability
space $(\Omega,\mathbb F, \cbra{\mathcal F_n}_n,\mathbb P)$, where $\mathcal F_n = \sigma(\mathcal X(k)|k\leq n)$, $k\in\mathbb Z_+$, is the natural filtration. 
The normal distribution with mean value $\mu$, and covariance matrix $\Sigma$ is denoted $\mathcal{N}(\mu, \Sigma)$.
The indicator function of the event $\sbra{\mathcal X\in S}$ is denoted $\mathds{1}_{\sbra{\mathcal X\in S}}$ 
and $\otimes$ denotes the Kronecker product.
Finally, ``$\min$'' (resp. ``$\max$'') defines the minimization (resp. maximization) operator while ``$\minimize$'' (resp. ``$\maximize$'') defines a minimization (resp. maximization) problem.



\section{Switched and Piecewise Affine Systems}
\label{Sec:SwitchedPreliminaries}

A general discrete-time switched system is described by:
\begin{equation}
\begin{aligned}
    x_{t+1} &= f_{\sigma_t}(x_t, u_t) + w_t \\
    y_{t} &= g_{\sigma_t}(x_t, u_t) + v_t,\quad t\in\mathbb{Z}_{+}
\end{aligned}
\label{eq:ss}
\end{equation}
where $x_t\in \mathbb{R}^n$ is the state vector of the system, 
$u_t\in \mathbb{R}^p$ the input, $y_t\in \mathbb{R}^q$ the output,
and $w_t\in \mathbb{R}^n$ and $v_t\in \mathbb{R}^q$ are noise terms.
The signal $\sigma_t\in \cbra{1, \ldots, s}$ 
defines the mode which is active at time $t$.
System \eqref{eq:ss} is a switched affine system when it 
can be expressed as:
\begin{equation}
\begin{aligned}
    x_{t+1} &= A_{\sigma_t} x_t + B_{\sigma_t} u_t + \bar f_{\sigma_t} + w_t \\
    y_{t} &= C_{\sigma_t} x_t + D_{\sigma_t} u_t + \bar g_{\sigma_t} + v_t,\quad t\in\mathbb{Z}_{+}.
\end{aligned}
\label{eq:sass}
\end{equation}
The matrices $A_i\in \mathbb{R}^{n\times n}$, $B_i\in \mathbb{R}^{n\times p}$,
$C_i\in \mathbb{R}^{q\times n}$, $D_i\in \mathbb{R}^{q\times p}$, 
$\bar f_i\in \mathbb{R}^{n}$, and $\bar g_i\in \mathbb{R}^{q}$ define the affine dynamics for 
each mode $i\in \cbra{1, \ldots, s}$.
%
System \eqref{eq:sass} is PWA when $\sigma_t$ 
is defined according to a polyhedral partition of the state and input space, i.e., when
\begin{equation}
    \sigma_t = i \iff \begin{bmatrix} x_t \\ u_t \end{bmatrix} \in R_i\subset R,
\end{equation}
where $R_i$, $i=1,\ldots,s$, are convex polyhedra defining a 
partition of the state-input domain $R\subseteq\mathbb{R}^{n+p}$, that is when $R_i\cap R_j = \emptyset$ for $i\neq j$, and $\bigcup_i R_i = R$.

Switched affine systems can be expressed in input--output form as 
(SARX) systems of fixed orders $n_a$, $n_b$, such that
for every component $y_t^{(i)}\in \mathbb R$ of the output vector $y_t\in \mathbb R^q$
it holds:
\begin{equation}
\begin{aligned}
    y_t^{(i)} = \bar\theta_{\sigma_t}^{(i)\T} \begin{bmatrix} r_t \\ 1 \end{bmatrix} + \bar e_t^{(i)},\
    i = 1,\ldots, q,
\end{aligned}
    \label{eq:sarx}
\end{equation}
%
where the regressor vector $r_t\in\mathbb{R}^{\bar d}$, $\bar d=qn_a+p(n_b+1)$, is defined by
\begin{equation}
    r_t = [y_{t-1}^\T \ldots y_{t-n_a}^\T u_{t}^\T u_{t-1}^\T \ldots u_{t-n_b}^\T]^\T\in\mathbb{R}^{\bar d}.
    \label{eq:r-vector}
\end{equation}
The parameter vectors $\bar\theta_j^{(i)}\in\mathbb{R}^{\bar d+1}$, $j\in\cbra{1,\ldots, s}$, 
define each ARX mode, and $\bar e_t\in\mathbb{R}^q$ is a noise term. 
%
%
%
Similarly, \eqref{eq:sarx} is PWARX if 
\begin{equation}
    \sigma_t = i \iff r_t \in P_i\subset P\subseteq \mathbb{R}^{\bar d},
    \label{eq:sigma-pwarx}
\end{equation}
and $\cbra{P_i}_{i=1}^s$ define a polyhedral partition
of $P\subseteq \mathbb{R}^{\bar d}$.

\subsection{Realization and Identification of PWARX Models}
\label{sSec:pwarx-identifiability}



Every observable switched affine system admits a SARX representation \cite{weiland2006equivalence}. 
Necessary and sufficient conditions for input--output realization of SARX and PWARX 
systems are given in \cite{paoletti2008necessary}, and \cite{paoletti2009input}, respectively. 
%
%
It is worth mentioning, however, that the number of modes and parameters 
can grow considerably when a PWA state-space system is converted into a 
minimum-order equivalent PWARX representation \cite{paoletti2009input}.
%


%
In spite of the increasing attention received by SARX and PWARX system identification, 
there are currently only few results on the identifiability of these systems
\cite{liberzon2003switching,paoletti2007identification,petreczky2010identifiability}.
%
%
%
%
The general identification problem for a PWARX system of the form \eqref{eq:sarx}-\eqref{eq:sigma-pwarx} 
can be formulated as a stochastic optimization problem over the parameters 
$\cbra{n_a, n_b, s, \cbra{\theta_i}_{i=1}^s, \cbra{P_i}_{i=1}^s}$.
We make the following assumption that will allow us to concentrate on the 
properties of PWARX identification, assuming known $(\tilde n_a, \tilde n_b)$ 
subject to potential computational bounds.


\begin{assumption}
Upper bounds $(\tilde n_a,\tilde n_b)$ on the orders of the model $(n_a,n_b)$ are known.
\label{as:orders}
\end{assumption}

\subsection{Realization and Identification of PWA State-Space Models}
\label{sSec:pwa-identifiability}



The problem of identifying a state-space representation of a switched affine system can be challenging. 
%
%
In particular, identifiability issues arise regarding 
the characterization of minimality of discrete-time switched linear systems \cite{mavridis2024state,vidal2002observability,petreczky2013realization}.
In this work, to ensure uniqueness of the realizations, given that all subsystems $i\in \cbra{1, \ldots, s}$ share the same state space, 
and simplify the presentation of our methodology, we make the following assumptions.

\begin{assumption}
$C_i = C$, $\forall i\in \cbra{1, \ldots, s}$ in system \eqref{eq:sass}. 
\label{as:C}
\end{assumption}
\begin{assumption}
We assume no affine dynamics, i.e., $\bar f_{\sigma_t}=0$, $\bar g_{\sigma_t}=0$, 
no feed-forward terms, i.e., $D_{\sigma_t}=0$, full observability, i.e., $C=I_n$,
and same zero-mean statistics for the error terms $w_t$ and $v_t$ for every mode of the system.
\label{as:simplifications}
\end{assumption}
%
Assumption \ref{as:C} implies that the order $n$ is known (observed) and enforces that the set of observations 
is acquired using the same observation mechanism, which leads to the realization of \eqref{eq:sass} being unique.
%
%
Assumptions \ref{as:simplifications} simplify the presentation of the proposed methodology without loss of generality. 
Together, Assumptions \ref{as:C} and \ref{as:simplifications} allow for the joint modeling of PWARX and state-space PWA systems, as defined in Section \ref{Sec:formulation}.

In addition to the realizations of the local systems being non-unique,
minimality and identifiability of the switched system does not necessarily imply that of the local subsystems
\cite{petreczky2010identifiability}.
In Theorem \ref{thm:identification}, we describe the conditions under which the local linear models of \eqref{eq:sass} (under Assumptions \ref{as:C}--\ref{as:simplifications})
can be identified, even when a subset of them is not controllable (minimal) in isolation.
%
%

%
\begin{theorem}
Consider a bounded-input bounded-output linear discrete-time system of the form: 
\begin{equation}
\begin{aligned}
    x_{t+1} &= A x_t + B u_t,\quad t\in\mathbb{Z}_{+} \\
    y_t &= x_t,
\end{aligned}
    \label{eq:linear}
\end{equation}
where $x_t\in \mathbb{R}^n$, 
$u_t\in \mathbb{R}^p$, $A\in \mathbb{R}^{n\times n}$, and $B\in \mathbb{R}^{n\times p}$.
Denote $r_t = [x_t^\T u_t^\T]^\T$.
Then, if there exist some $\alpha,\beta, T>0$ such that 
\begin{equation}
    \alpha I_{n+p} \preceq \sum_{\tau=t}^{t+T} r_t r_t^\T \preceq \beta I_{n+p},\quad \forall t\geq 0,
    \label{eq:PE}
\end{equation}
the augmented parameter matrix $\hat \Theta_t = [\hat A_{t} | \hat B_{t}]$ updated by the recursion 
\begin{equation}
    \hat \Theta_{t+1} = \hat \Theta_t - \gamma \pbra{\hat \Theta_t r_t - x_{t+1}} r_t^\T,\quad t\geq 0,
    \label{eq:theta_recursion}
\end{equation}
for some $\gamma>0$, asymptotically converges to $\Theta = [A | B]$. 
\label{thm:identification}
\end{theorem}
\begin{proof}
See Appendix \ref{App:identification}.
\end{proof}

\noindent
As a result of Theorem \ref{thm:identification}, throughout this paper, we make the following assumption to ensure identifiability of \eqref{eq:sass} under Assumptions \ref{as:C}--\ref{as:simplifications}:
\begin{assumption}
All linear subsystems $i\in\cbra{1,\ldots,s}$ of \eqref{eq:sass} are asymptotically bounded, and 
the bounded control input $u_t$ is designed such that for every mode $i\in\cbra{1,\ldots,s}$ of \eqref{eq:sass},
there exist some $\alpha_i,\beta_i, T_i>0$ for which the following 
persistence of excitation condition holds:
\begin{equation}
    \alpha_i I_{n+p} \preceq \sum_{\tau=t}^{t+T_i} \begin{bmatrix} x_\tau x_\tau^\T & x_\tau u_\tau^\T \\ 
                                                  u_\tau x_\tau^\T & u_\tau u_\tau^\T \end{bmatrix} \preceq \beta_i I_{n+p},\ \forall t\geq 0.
    \label{eq:PE-as}
\end{equation}
\label{as:PE}
\end{assumption}
\begin{remark}
    Informally, condition \eqref{eq:PE-as} states that not every subsystem in \eqref{eq:sass} should be controllable (minimal), 
    as long as the boundaries of each mode (region $R_i$ in the state-input system) 
    are visited often enough. 
\end{remark}
\begin{remark}
    The assumption of asymptotic boundedness and controllability (thus, minimality) for all subsystems of \eqref{eq:sass} 
    would simplify the condition \eqref{eq:PE-as} to a persistence of excitation criterion for the input $u_t$ for each subsystem 
    separately.
    However, it is a limiting assumption in a practical sense. 
\end{remark}
%

\section{Switched System Identification as an Optimization Problem}
\label{Sec:formulation}

Consider a switched linear system of the form: 
\begin{equation}
\begin{aligned}
    \psi_{t} &= \Theta_{i} \phi_t + e_t,\\
    &= [\phi_t^\T \otimes I_{m}] \theta_i + e_t, \text{ if } \phi_t\in S_i,\ t\in \mathbb Z_+,\\
\end{aligned}
    \label{eq:pwa-theta}
\end{equation}
where $\psi_t \in \mathbb R^{m}$, $\phi_t\in \mathbb R^{d}$, 
$\sigma_t\in \cbra{1,\ldots, s}$,
$\Theta_i\in \mathbb R^{m\times d}$, for all $i=1,\ldots, s$, 
$\theta_i=\text{vec}(\Theta_i)\in\mathbb R^{md}$,
$e_t \in \mathbb R^{m}$ is a zero-mean noise signal, and 
$\cbra{S_i}_{i=1}^s$ define a polyhedral partition
of $S\subseteq \mathbb{R}^{d}$.
%
System \eqref{eq:sarx} can be written in the form \eqref{eq:pwa-theta} with 
$\psi_t = y_t \in \mathbb R^q$, $\phi_t = [r_t^\T 1]^\T \in \mathbb R^{\bar d+1}$, and 
$\Theta_i = [\bar\theta_i^{(1)} \ldots \bar\theta_i^{(q)}]^\T$, where $m=q$, and $d=\bar d+1$.
In addition, 
system \eqref{eq:sass} under Assumptions \ref{as:C}, \ref{as:simplifications} 
can be written in the form \eqref{eq:pwa-theta} with 
$\psi_t = x_{t+1} \in \mathbb R^n$, $\phi_t = [x_t^\T u_t^\T]^\T \in \mathbb R^{n+p}$, and 
$\Theta_i = [A_{i} | B_{i}]$, where $m=n$, and $d=n+p$.
Notice that, in this case, \eqref{eq:pwa-theta} holds for 
$(t-1)\in \mathbb Z_+$, i.e., $t\in\mathbb N$.

Under the identifiability conditions discussed in Section \ref{Sec:SwitchedPreliminaries},
the general identification problem for a switching system 
of the form \eqref{eq:pwa-theta} can be
formulated as a stochastic optimization problem over the parameters 
$\cbra{s, \cbra{\theta_i}_{i=1}^s, \cbra{S_i}_{i=1}^s}$, 
as follows:
\begin{equation}
    \minimize_{s,\cbra{\theta_i},\cbra{S_i}} \
    \E{\sum_{i=1}^s \mathds 1_{\sbra{\Phi\in S_i}}  d_\rho \pbra{\Psi, [\Phi^\T \otimes I_{m}] \theta_i}
    },
    \label{eq:id-problem}
\end{equation}
where $\Psi\in\mathbb R^m$ and $\Phi\in\mathbb R^{d}$ represent random variables,
realizations of which constitute the system observations, 
the nonnegative measure $d_\rho$
is an appropriately defined dissimilarity measure, and 
the expectation is taken with respect to the joint distribution of $(\Psi,\Phi)\in \mathbb R^{m+d}$
that depends on the system dynamics, the control input, and the noise term in \eqref{eq:pwa-theta}.
%

It is clear that the optimization problem \eqref{eq:id-problem} is computationally hard and becomes intractable as the number of modes and states increases. 
In particular, the number of modes $s$ is unknown and
completely alters the cardinality and the domain of the set of parameter vectors
$\cbra{\theta_i}_{i=1}^s$ that represent the dynamics of the system.
In addition, a parametric representation for the polyhedral regions $\cbra{S_i}$ 
should be defined. 
%
%

%

To represent the regions $\cbra{S_i}$, we will follow a Voronoi tessellation approach based on prototypes.
We introduce a set of parameters $\hat\phi \coloneq \cbra{\hat\phi_i}_{i=1}^K$, $\hat\phi_i\in S$ and
define the regions:
\begin{equation}
\Sigma_i = \cbra{\phi\in S: i=\argmin_j  d_\rho (\phi,\hat\phi_j)},\ i=1,\ldots,K.    
\label{eq:Sigma-partition}
\end{equation}
The measure $d_\rho$ can be designed 
such that the Voronoi regions $\Sigma_i$
are polyhedral, e.g., when $d_\rho$ is a squared Euclidean distance or any Bregman divergence, as will be explained in Section \ref{sSec:oda}.
In this sense, each $S_i$ can be mapped to a region $\Sigma_j$ (for $K=s$) or the union of a subset of $\cbra{\Sigma_j}$ (for $K>s$), according to a predefined rule, 
as will be explained in Section \ref{sSec:modes}. 
An illustration of this partition is given in Fig. \ref{fig:pwa-example}.

In addition to the prototype parameters $\cbra{\hat\phi_i}_{i=1}^K$, we also introduce 
a set of parameters $\hat\theta \coloneq \cbra{\hat\theta_i}_{i=1}^{K}$, $\hat\theta_i\in\mathbb R^{md}$, with 
each $\hat\theta_i$ associated with the region $\Sigma_i$ according to \eqref{eq:Sigma-partition}.
Representing the augmented random vector 
\begin{equation}
X = \begin{bmatrix} \Psi\\ \Phi \end{bmatrix}\in \Pi \subseteq \mathbb R^{m+d},
\label{eq:X}
\end{equation}
we can define a set of augmented codevectors $\mu\coloneq \cbra{\mu_i}_{i=1}^K$ as 
\begin{equation}
\mu_i = \begin{bmatrix} z(\phi,\hat\theta_i)\\ \hat\phi_i \end{bmatrix}\in \Pi,\ 
i=1,\ldots,K,
\label{eq:mu-definition}
\end{equation}
where the first component of each $\mu_i$%
\footnote{Throughout this paper we will use the notation $\mu_i$, 
$\mu_i(\hat\phi_i)$, $\mu_i(\hat\theta_i)$, $\mu_i(\hat\theta_i,\hat\phi_i)$, $\mu_i(\phi,\hat\theta_i,\hat\phi_i)$ interchangeably, to showcase 
the dependence on the variables of interest in each case.}
is a mapping $z(\phi,\hat\theta_i) = [\phi^\T \otimes I_{m}] \hat\theta_i$
that simulates the local model dynamics in \eqref{eq:pwa-theta}
with unknown parameters $\theta_i$, and the second component is a set of unknown codevectors $\hat\phi_i$ that define the partition in 
\eqref{eq:Sigma-partition}.

Problem \eqref{eq:id-problem} can then be decomposed into two interconnected stochastic optimization problems.
Assuming $\cbra{\hat\theta_i}_{i=1}^K$ are known, the optimization problem 
\begin{equation}
    \minimize_{\hat\phi} \
    \E{\sum_{i=1}^K \mathds 1_{\sbra{\Phi\in \Sigma_i(\hat\phi)}}  d_\rho \pbra{X(\Psi,\Phi), \mu_i(\hat\theta_i,\hat\phi_i)}
    } 
    \label{eq:clustering-problem}
\end{equation}
finds the optimal parameters $\cbra{\hat\phi_i}_{i=1}^K$ 
that define the partition $\cbra{\Sigma_i}_{i=1}^K$
subject to the joint distribution of $(\Psi,\Phi)$, 
and is, therefore, a mode switching signal identification problem.

On the other hand, assuming the partition $\cbra{\Sigma_i}_{i=1}^K$ (and, therefore, $\cbra{S_i}_{i=1}^s$) 
is known, the optimization problem 
\begin{equation}
    \minimize_{\hat\theta} \
    \E{\sum_{i=1}^K \mathds 1_{\sbra{\Phi\in \Sigma_i}}  d_\rho \pbra{\Psi, [\Phi^\T \otimes I_{m}] \hat\theta_i}
    } 
    \label{eq:filtering-problem}
\end{equation}
is a system identification problem for each mode of the system.

In Section \ref{Sec:ODA} we address the question of finding the optimal number $K$ 
according to a performance-complexity trade-off, as well as finding a mapping between 
$\cbra{\Sigma_i}_{i=1}^K$ and $\cbra{S_i}_{i=1}^{\hat s}$ for the lowest possible number $\hat s\geq s$.
In Section \ref{Sec:PWAID} we tackle the problem of 
estimating both $\hat\phi$ and $\hat\theta$ by
solving \eqref{eq:clustering-problem} and \eqref{eq:filtering-problem} 
as a system of interconnected stochastic optimization problems in real-time using principles from two-timescale stochastic approximation theory.
\begin{figure}[t]
\centering
\includegraphics[trim=0 0 0 0,clip,width=0.45\textwidth]{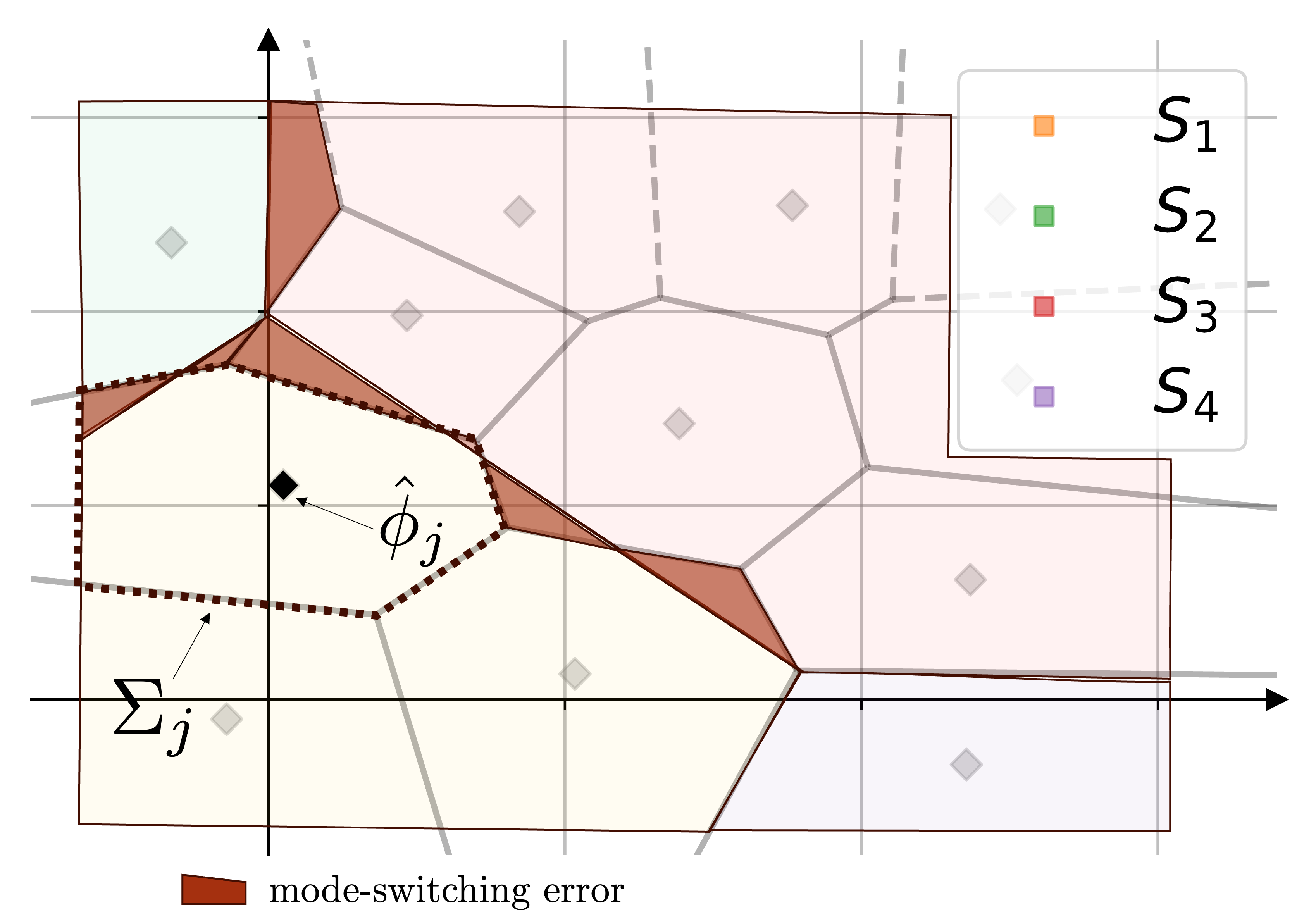}
\caption{Illustration of the partition $\cbra{S_i}_{i=1}^s$ ($s=4$) of the state-input space $S$ and its connection to the artificial partition $\cbra{\Sigma_j}_{j=1}^K$ ($K>s$). The optimal parameters $\cbra{\hat\phi_j}$ induce a partition $\cbra{\Sigma_j}$ that minimizes the mode switching error.}
\label{fig:pwa-example}
\vspace{-1em}
\end{figure}
%

\section{Mode Identification with Online Deterministic Annealing}
\label{Sec:ODA}


We aim to construct a recursive stochastic optimization algorithm to solve problem 
\eqref{eq:clustering-problem} while progressively estimating the number $K$ of the 
augmented codevectors $\cbra{\mu_i}_{i=1}^K$, an estimate $\hat s$ of the actual number of modes, 
and a mapping between $\cbra{\Sigma_i}_{i=1}^K$ and $\cbra{S_i}_{i=1}^{\hat s}$.
Recall that the observed data 
are represented by the random variable 
$X\in \Pi$ in \eqref{eq:X}, and
the augmented codevectors
$\cbra{\mu_i}_{i=1}^K$ are normally treated as constant parameters to be estimated.
To progressively estimate $K$ and $\hat s$, we will adopt the online deterministic annealing approach 
\cite{mavridis2023online,mavridis2023annealing}, 
and define a probability space over an arbitrary number of codevectors, 
while constraining their distribution using a maximum-entropy principle 
at different levels.
First we define
a quantizer $Q:\Pi \rightarrow \Pi$ 
as a stochastic mapping of the form:
\begin{equation}
Q(x) = \mu_i\ \text{ with probability } p(\mu_i|x).
\end{equation}
Then we formulate
the multi-objective optimization
\begin{equation}
\minimize_{\hat\phi} \ F_\lambda(\mu) = (1-\lambda) D(\mu) - \lambda H(\mu),\ \lambda\in[0,1),
\label{eq:F}
\end{equation}
where the dependence on $\hat\phi$ comes through $\mu(\hat\phi)$, the term
\begin{equation}
\begin{aligned}
D(\mu) &= \E{d_\rho\pbra{X,Q}}
=\int p(x) \sum_i p(\mu_i|x) d_\rho(x,\mu_i) ~\textrm{d}x
\end{aligned}
\label{eq:D}
\end{equation}	
is a generalization of the objective in (\ref{eq:clustering-problem}), and 
\begin{equation}
\begin{aligned}
H(\mu) &= \E{-\log P(X,Q)} 
\\&
=H(X) - \int p(x) \sum_i p(\mu_i|x) \log p(\mu_i|x) ~\textrm{d}x
\end{aligned}    
\label{eq:H}
\end{equation}
is the Shannon entropy.
This is now a problem of finding the locations $\cbra{\hat\phi_i}$ 
and the corresponding probabilities
$\cbra{p(\mu_i|x)=\mathbb P[Q=\mu_i|X=x]}$.

Notice that, for $p(\mu_i|x)=\mathds 1_{\sbra{\phi\in \Sigma_i(\hat\phi)}}$ and $\lambda=0$, 
\eqref{eq:F} is equivalent to \eqref{eq:clustering-problem}.
In that sense, \eqref{eq:F} introduces extra optimization parameters in the probabilities $\cbra{p(\mu_i|x)}$, and 
the parameter $\lambda$ that defines a homotopy $F_\lambda$.
However, the advantages of this approach are notable, and, perhaps counter-intuitively,
lead to numerical optimization solutions with several computational benefits.
On the one hand, the Lagrange multiplier $\lambda\in[0,1)$ controls the trade-off between 
$D$ and $H$, which, as will be shown, is a trade-off between performance and complexity.
%
On the other hand, the use of the conditional probabilities $\cbra{p(\mu_i|x)}$ allows for the definition of the entropy term $H$, 
which introduces several useful properties 
\cite{mavridis2023online,mavridis2023annealing,mavridis2021maximum,mavridis2021progressive,mavridis2020detection}.
In particular, as we will show in Section \ref{sSec:bifurcation}, 
reducing the values of $\lambda$ defines 
a direction that resembles an annealing process
\cite{mavridis2023online,rose1998deterministic} and induces
a bifurcation phenomenon, 
with respect to which, the number of unique codevectors $K_\lambda$ depends on $\lambda$ 
and is finite for any given value of
$\lambda>0$.
This process also introduces robustness with respect to initial conditions
\cite{mavridis2023online,mavridis2022risk}.
%

\subsection{Solving the Optimization Problem}
\label{sSec:oda}

To solve \eqref{eq:F} for a given value of $\lambda$, 
we successively minimize $F_\lambda$
first with respect to the 
association probabilities $\cbra{p(\mu_i|x)}$, 
and then with respect to the codevector locations $\mu$.
The solution of the optimization problem
\begin{equation}
\begin{aligned}
F_\lambda^*(\mu) &= \min_{\cbra{p(\mu_i|x)}} F_\lambda(\mu),\\
\quad\quad
\text{s.t.} & \sum_i p(\mu_i|x) = 1,
\end{aligned}
\label{eq:Fstar}
\end{equation}
is given by the Gibbs distributions \cite{mavridis2023multi}:
\begin{equation}
p^*(\mu_i|x) = \frac{e^{-\frac{1-\lambda}{\lambda}d_\rho(x,\mu_i)}}
			{\sum_j e^{-\frac{1-\lambda}{\lambda}d_\rho(x,\mu_j)}} ,~ \forall x\in \Pi.
\label{eq:gibbs}
\end{equation}
In order to minimize $F^*(\mu)$ with respect to $\hat\phi$ 
we set the gradients to zero
\begin{equation}
\frac \dd {\dd{\hat\phi}} F_\lambda^*(\mu) = \frac \dd {\dd\mu} F_\lambda^*(\mu) \frac {\dd \mu} {\dd{\hat\phi}} = 0 
\label{eq:dFdphi}
\end{equation}
where 
$\frac {\dd \mu} {\dd{\hat\phi}}=\begin{bmatrix} 0_{m\times d} \\ I_d \end{bmatrix}$, and 
\begin{equation}
\begin{aligned}
&\frac \dd {\dd{\mu}} F_\lambda^*(\mu) =
\frac \dd {\dd\mu} \pbra{ (1-\lambda)D(\mu) - \lambda H(\mu) }  
\\&
=
\sum_i \int p(x) p^*(\mu_i|x) \frac \dd {\dd\mu_i} d_\rho(x,\mu_i) ~ \dd x  = 0,
\end{aligned}
\label{eq:dFdmu}
\end{equation}
where we have used (\ref{eq:gibbs}) and direct differentiation with similar arguments as in \cite{mavridis2023multi}.
It follows that
$\frac \dd {\dd\hat\phi} F_\lambda^*(\mu) = 0 $
which implies that
\begin{equation}
\begin{aligned}
\int p(x) p^*(\mu_i|x) \frac \dd {\dd \mu_i} d_\rho(x,\mu_i) ~\dd x 
\begin{bmatrix} 0_{m\times d} \\ I_d \end{bmatrix}
= 0,\ \forall i.
\end{aligned}
\label{eq:M}
\end{equation}
Equation
(\ref{eq:M}) has a
closed-form solution if the dissimilarity measure $d_\rho$ 
belongs to the family of Bregman divergences
\cite{banerjee2005clustering,mavridis2023online},%
information-theoretic dissimilarity measures that 
include the squared Euclidean distance and 
the Kullback-Leibler divergence, and are defined as follows:

%
\begin{definition}[Bregman Divergence]
	Let $ \rho: S \rightarrow \mathbb{R}$, 
	be a strictly convex function defined on 
	a vector space $S\subseteq \mathbb{R}^d$ such that $\phi$  
	is twice F-differentiable on $S$. 
	The Bregman divergence 
	$d_{\rho}:H \times S \rightarrow \left[0,\infty\right)$
	is defined as:
	\begin{align*}
		d_{\rho} \pbra{x, \mu} = \rho \pbra{x} - \rho \pbra{\mu} 
							- \pder{\rho}{\mu} \pbra{\mu} \pbra{x-\mu},
	\end{align*}
	where $x,\mu\in S$, and 
 the continuous linear map 
	$\pder{\rho}{\mu} \pbra{\mu}: S \rightarrow \mathbb{R}$ 
	is the Fr\'echet derivative of $\rho$ at $\mu$.
	\label{def:BregmanD}
\end{definition}

Throughout this manuscript, we will assume that 
the dissimilarity measure $d_\rho$ in \eqref{eq:Sigma-partition} is a Bregman divergence, and, in particular, the squared Euclidean distance.
Then 
the solution to the optimization problem
\begin{equation}
    \minimize_{\hat\phi}~ F_\lambda^*\pbra{\mu(\hat\phi)},
    \label{eq:minFstar}
\end{equation}
where $F_\lambda^*(\mu)$ is the solution of \eqref{eq:Fstar} 
for a given $\lambda\in [0,1)$ and 
$p^*(\mu_i|x)$ is given by \eqref{eq:gibbs}, 
is given by Theorem \ref{thm:bregman_in_DA}.

\begin{theorem}
If $d_\rho:\Pi\times\Pi\rightarrow \mathbb R_+$ is a Bregman divergence,
then
\begin{equation}
\hat\phi_i^* 
= \frac{\int \phi p(x) p^*(\mu_i|x) ~\dd x}{p^*(\mu_i)}
\label{eq:mu_star}
\end{equation}
is a solution to the optimization problem \eqref{eq:minFstar}.
%
\label{thm:bregman_in_DA}
\end{theorem}
\begin{proof}
    By definition, for a Bregman divergence $d_\rho:\Pi\times\Pi\rightarrow \mathbb R_+$ based on 
    a strictly convex function $\rho:\Pi\rightarrow \mathbb R$, it holds that 
    $\pder{d_{\rho}}{\mu}(x,\mu) = - \abra{\nabla^2 \rho(\mu),(x-\mu)}$. Similar to \cite{mavridis2023annealing}, with standard algebraic manipulations, \eqref{eq:M} then becomes
    \begin{equation}
    \int (\phi-\hat\phi_i^*) p(x) p^*(\mu_i|x) ~\dd x= 0,\ \forall i,
    \label{eq:bregman_proof_int}
    \end{equation}
    where $p^*(\mu_i|x)$ is given by  
    \eqref{eq:gibbs} and the integral is defined over the domain $\Pi$.
    Eq. \eqref{eq:bregman_proof_int} is equivalent to 
    (\ref{eq:mu_star}) since 
    $\int p(x) p^*(\mu_i|x) ~\dd x = p^*(\mu_i)$.
\end{proof}

\begin{remark}
    The partition $\cbra{\Sigma_i}$ induced by \eqref{eq:Sigma-partition}
    and a dissimilarity measure $d_\rho$ that belongs to the family of Bregman divergences,
    is separated by hyperplanes \cite{banerjee2005clustering}.
    As a result, each $\Sigma_i$ is a polyhedral region
    for a bounded domain $S$.
    \label{rmk:polyhedral}
\end{remark}


%
Based on Theorem \ref{thm:bregman_in_DA},
Theorem \ref{thm:ODA} below constructs a gradient-free stochastic approximation algorithm 
that recursively estimates (\ref{eq:mu_star}).

\begin{theorem}
The sequence $\hat\phi_i(t)$ constructed by the recursive updates
\begin{equation}
\begin{cases}
\hat\rho_i(t+1) &= \hat\rho_i(t) + \beta(t) \sbra{ \hat p_i(t) - \hat\rho_i(t)} \\
\sigma_i(t+1) &= \sigma_i(t) + \beta(t) \sbra{ \phi_t \hat p_i(t) - \sigma_i(t)},
\end{cases}
\label{eq:oda_learning1}
\end{equation}
where $x_t = [\psi_t^\T \phi_t^T]^\T$ represents external input
with $\psi_t\sim \Psi$, $\phi_t\sim \Phi$, $\sum_t \beta(t) = \infty$, $\sum_t \beta^2(t) < \infty$,
and the quantities $\hat p_i(t)$ and $\hat\phi_i(t)$ 
are recursively updated 
as follows:
\begin{equation}
\begin{aligned}
\hat\phi_i(t) = \frac{\sigma_i(t)}{\hat\rho_i(t)},\quad
\hat p_i(t) = \frac{\hat\rho_i(t) e^{-\frac{1-\lambda}{\lambda}d_\rho(x_t,\mu_i(t))}}
			{\sum_i \hat\rho_i(t) e^{-\frac{1-\lambda}{\lambda}d_\rho}(x_t,\mu_i(t))}, 
\end{aligned}
\label{eq:oda_learning2}
\end{equation}
with $\mu_i(t)=[z_i^\T(\phi_t,\hat\theta_i),\hat\phi_i(t)^\T]^\T$,
converges almost surely to $\hat\phi_i^*$ given in \eqref{eq:mu_star}.
\label{thm:ODA}
\end{theorem}
\begin{proof}
    The proof follows similar arguments as Theorem $5$ of 
    \cite{mavridis2023annealing}.
\end{proof}

\begin{remark}
Intuitively, the quantity $\hat \rho_i$ in \eqref{eq:oda_learning1} is an estimate of the probability $p(\mu_i)$. In that sense, $\sigma_i$ becomes an estimate of $\E{\mathds{1}_{\cbra{\mu}}\Phi}$, and 
$\hat\phi_i$ becomes an estimate of $\E{\Phi|\mu}$, which is a generalization of the centroid form found in clustering algorithms \cite{mavridis2023online}.

\end{remark}

\begin{remark}
    Notice that the dynamics of   
    \eqref{eq:oda_learning1} 
    can be expressed as: 
    {\small
    \begin{equation}
        \begin{aligned}
            \hat\phi_i(t+1) &= 
            \frac{\beta(t)}{\hat\rho_i(t)}\bigg[\frac{\sigma_i(t+1)}{\hat\rho_i(t+1)}
            \pbra{\hat\rho_i(t)-\hat p_i(t)}
            + \phi_t \hat p_i(t)-\sigma_i(t)\bigg],
        \end{aligned}
        \label{eq:mu-updates}
    \end{equation}
    }%
    where the recursive updates take place for every codevector $\hat\phi_i$
    sequentially. This is a discrete-time dynamical system 
    that presents bifurcation phenomena with respect to the parameter $\lambda$, 
    i.e., the number of equilibria of this system changes with respect to the
    value $\lambda$ which is hidden inside the term $\hat p_i(t)$ in 
    \eqref{eq:oda_learning2}.
    According to this phenomenon, the number of distinct values of $\hat\phi_i$ is finite, 
    and the updates need only be taken with respect to these values that we call 
    ``effective codevectors''.
    This is discussed in Section \ref{sSec:bifurcation}.
\end{remark}

\subsection{Bifurcation Phenomena}
\label{sSec:bifurcation}

In Section \ref{sSec:oda} we described how to solve the optimization 
problem for a given value of the parameter $\lambda$.
The main idea of the proposed approach is to solve 
a sequence of optimization problems of the form \eqref{eq:F} with decreasing values 
of $\lambda$.
This process then becomes a homotopy optimization method \cite{lin2023continuation}.
In particular, the usage of the entropy term resembles annealing optimization methods and grants $\lambda$ the name of a 'temperature' parameter. 
Notice that, so far, we have assumed an arbitrary number of codevectors $K$.
We will show that the unique values of the set $\cbra{\hat\phi_i}$ 
that solves \eqref{eq:F},
form a finite set of $K_\lambda$ values that we will refer to as ``effective codevectors''.
%

Notice that at high temperature ($\lambda\rightarrow 1$),~(\ref{eq:gibbs}) yields
uniform association probabilities $p(\mu_i|x)=p(\mu_j|x),\ \forall i,j, \forall x$, 
and as a result of (\ref{eq:mu_star}), all pseudo-inputs are located at the same point
$\hat\phi_i = \Ep{X}{\phi},\ \forall i$,
which means that there is one unique ``effective'' codevector given by $\Ep{X}{\phi}$.
As $\lambda$ is lowered below a critical value, a bifurcation phenomenon occurs, 
when the number of ``effective'' codevectors increases, 
which describes an annealing process \cite{mavridis2023online,rose1998deterministic}.
Mathematically, this occurs when the existing solution $\hat\phi^*$ given by (\ref{eq:mu_star}) 
is no longer the minimum of the free energy $F^*$,
as the temperature $\lambda$ crosses a critical value.
Following principles from variational calculus,
we can track the bifurcation by the condition:
\begin{equation}
    \frac{d^2}{d\epsilon^2} F^*(\cbra{\hat\phi+\epsilon \hat\psi})\bigg|_{\epsilon=0} \geq 0,
    \label{eq:soc}
\end{equation}
for all choices of finite perturbations $\cbra{\hat\psi}$.
%
Using (\ref{eq:soc}) and direct differentiation, one can show 
that bifurcation depends on the temperature coefficient $\lambda$ 
(and the choice of the Bregman divergence, through the function $\rho$)
\cite{mavridis2023annealing,mavridis2023multi}.
%
%
%
In other words, the number of codevectors increases countably many times as 
the value of $\lambda$ decreases, resulting, at the limit, in a consistent density estimator \cite{mavridis2023multi}. 
However, an algorithmic implementation needs only
as many codevectors in memory as the number of ``effective'' codevectors.

In practice.
we can detect the bifurcation points 
by introducing perturbing pairs of codevectors at each 
temperature level $\lambda$.
In this way, the codevectors $\hat\phi$ are doubled by inserting a perturbation of each $\hat\phi_i$ in 
the set of effective codevectors.
The newly inserted codevectors will merge with their pair if 
a critical temperature has not been reached and separate otherwise.
The merging criterion takes the form:
\begin{equation}
\frac{1-\lambda}{\lambda} d_\rho(\hat\phi_i,\hat\phi_j)\leq \epsilon_n,\ \forall i,j,   
\label{eq:phi_condition}
\end{equation}
for a given threshold $\epsilon_n$.
The pseudocode for this algorithm is presented in Alg. \ref{alg:ODA}.
A detailed discussion on the implementation of the 
original online deterministic annealing algorithm,
its complexity,
and the effect of its parameters,
can be found in \cite{mavridis2023online,mavridis2023annealing,mavridis2023multi}.

\subsection{Estimating the number of modes}
\label{sSec:modes}

As illustrated in Fig. \ref{fig:pwa-example}, the problem formulation developed in 
Section \ref{Sec:formulation} defines a possibly imperfect surjective mapping from $\cbra{\Sigma_j}_{j=1}^K$ to 
$\cbra{S_i}_{i=1}^s$ such that each $S_i$ is defined as a union of a subset of $\cbra{\Sigma_j}_{j=1}^K$.
In this section, we define a recursive method to automatically construct this mapping, a critical addition to the methods proposed in \cite{mavridis2024state,mavridis2023identification}.

It is worth noting that the construction of $\cbra{\Sigma_j}_{j=1}^K$ defines a consistent density estimator of the mode swithcing behavior on $S$ in the limit $\lambda\rightarrow 0$ (which induces $K\rightarrow \infty$) \cite{mavridis2023multi}. However, according to Remark \ref{rmk:polyhedral}, it is possible for this mapping to be perfect even with bounded $K$ if $P$ is a polyhedral partition and the reconstruction is ideal. Then each $S_i$ is perfectly represented, inducing zero mode switching error.
In addition, the design of an appropriate termination criterion for Alg. \ref{alg:ODA} 
is an open question and is subject to the trade-off between the number
$K$ and the minimization of the identification error.
In this work, we make use of the condition $K\leq K_{\max}$ as a termination criterion,
where $K_{\max}$ represents the computational capacity of the identification device.

Recall that each $\Sigma_j$ is associated with a parameter vector $\hat\theta_j$, $j=1,\ldots, K$.
Assuming a set $\bar\theta = \cbra{\bar \theta_i}_{i=1}^{\hat s}$, we define each $\hat\theta_j$ as the mapping:
\begin{equation}
    \hat\theta_j(\bar\theta) = \bar\theta_i,\ \text{if } i=\argmin_k d_\rho(\hat\theta_j,\bar\theta_k).
    \label{eq:theta_rule}
\end{equation}
In this way $\Sigma_j\in S_i$ if $\hat\theta_j(\bar\theta) = \bar\theta_i$.
Therefore, given \eqref{eq:theta_rule}, the goal now is to find $\hat s$ and $\bar\theta$ such that 
$\hat s = s$, and $\bar\theta_i=\theta_i$, $\forall i\in \cbra{1,\ldots,s}$.
We follow a similar approach to the bifurcation mechanism described in Section \ref{sSec:bifurcation}.
Starting with one codevector $\hat\phi_0$, we define $\bar\theta_0=\hat\theta_0$.
Every time a codevector $\hat\phi_j$ is split into a pair of perturbed codevectors, 
a new $\hat\theta_{j^\prime}$ is introduced.
After convergence for a given $\lambda$, merging of the codevectors is detected by \eqref{eq:phi_condition}.
For the insertion of a new $\bar\theta_i$ we check the condition:
\begin{equation}
d_\rho(\hat\theta_j,\bar\theta_i)> \epsilon_s,\ \forall j,
\label{eq:theta_condition}
\end{equation}
with respect to a given threshold $\epsilon_s$.
Notice that in contrast to \eqref{eq:phi_condition}, \eqref{eq:theta_condition} does not depend on $\lambda$.
If \eqref{eq:theta_condition} is satisfied, a new $\bar\theta_i$ is introduced and $\hat s \leftarrow \hat s + 1$.
This process is integrated in the mode identification algorithm and its pseudocode is presented in Alg. \ref{alg:ODA}.

\begin{remark}
Note that $\cbra{\hat\theta_j}$ are only used as functions of $\bar\theta$, 
and the parameters $\cbra{\bar\theta_i}$ are the ones that are being updated by the local system identification algorithm 
that will be presented in Section \ref{Sec:PWAID}. 
\end{remark}

 \begin{algorithm}[hb!]
 \caption{Switched System Identification}
 \label{alg:ODA}
 \begin{algorithmic}
 \STATE Set parameters and initialize $\hat\phi=\cbra{\hat\phi_0},\bar\theta=\cbra{\bar\theta_0}$
 \WHILE{$K<K_{\textrm{max}}$ \textbf{and} $\lambda>\lambda_{\textrm{min}}$}
 \STATE Perturb  
 	$\hat\phi_i \gets  
 		\cbra{\hat\phi_i+\delta, \hat\phi_i-\delta}$, $\forall i$ 
 \STATE Set $t \gets 0$
 \REPEAT 
 %
 \STATE Observe $x=(\psi,\phi)$ according to \eqref{eq:pwa-theta}   
 \STATE Update $\bar\theta_w$, $w=\argmin_j  d_\rho (\phi,\hat\phi_j)$, using \eqref{eq:sgd-updates}
 \FOR{$i = 1,\ldots, K$} 
 \STATE Update $\hat\phi$ using \eqref{eq:oda_learning1}, \eqref{eq:oda_learning2} 
 \ENDFOR
 \STATE $t\gets t+1$
 \UNTIL Convergence 
 \STATE Discard $\hat\phi_i$ if $\frac{1-\lambda}{\lambda}d_\rho(\hat\phi_j,\hat\phi_i)<\epsilon_n$, 
	$\forall i,j,i\neq j$
 \STATE Insert $\hat\theta_i$ in $\bar\theta$ if   
 	$d_\rho(\hat\theta_j,\hat\theta_i)> \epsilon_s,\ \forall j$
 \STATE Lower temperature $\lambda \gets \gamma \lambda$, $0<\gamma<1$
 \ENDWHILE
 \STATE Define $\cbra{\Sigma_i}_{i=1}^K$ using \eqref{eq:Sigma-partition}
 \STATE Define $\hat s = \text{card}(\bar\theta)$
 \STATE Define $\cbra{S_i}_{i=1}^{\hat s}$ by $\Sigma_j\in S_i$ if $\hat\theta_j(\bar\theta) = \bar\theta_i$
 \STATE Estimated Model Parameters: $\hat s$, $\cbra{S_i}_{i=1}^{\hat s}$, $\cbra{\bar\theta_i}_{i=1}^{\hat s}$
        
 \end{algorithmic}
 \end{algorithm}

\section{Piecewise Affine System Identification}
\label{Sec:PWAID}

In this section we review standard recursive system identification for 
estimating the parameters $\cbra{\bar\theta_i}$ of the local models given 
knowledge of the partition $\cbra{S_i}$.

We show that this kind of recursive identification can be formulated as a stochastic approximation algorithm, and that it can be combined using the theory of two-timescale stochastic approximation with the stochastic approximation method of Section \ref{Sec:ODA} to estimate both 
$\cbra{S_i}$ and $\cbra{\bar\theta_i}$ at the same time.

\subsection{Recursive Identification of Local Models}
\label{sSec:gradient-descent}

Recall that, given knowledge of the partition $\cbra{S_i}_{i=1}^s$, each local linear model of the PWA system in \eqref{eq:pwa-theta} is 
completely defined by the parameters $\cbra{\theta_i}$.
In the following, we develop 
a stochastic approximation recursion to estimate  $\cbra{\bar\theta_i}$.
First we define the error:
\begin{equation}
    \epsilon(t) =  \sum_i \mathds{1}_{\sbra{\phi_t\in S_i}} [\phi_t^\T \otimes I_{m}] \bar\theta_i  - \psi_t 
    \label{eq:epsilon}
\end{equation}
A stochastic gradient descent approach aims to minimize the error:
\begin{equation}
\minimize_{\bar\theta_i}~ \frac 1 2 \E{ \|\epsilon(t) \|^2 },
\label{eq:sgd-min}
\end{equation}
using the recursive updates:
\begin{equation}
\begin{aligned}
\bar\theta_i(t+1) &= \bar\theta_i(t) - \alpha(t)  \pbra{\nabla_{\bar\theta_i} \epsilon(t)} \epsilon(t) \\  
&= \bar\theta_i(t) - \alpha(t) [\phi_t^\T \otimes I_{m}]^\T \epsilon(t)
\end{aligned}
\label{eq:sgd-updates}
\end{equation}
where $\sum_n \alpha(n) = \infty$, $\sum_n \alpha^2(n) < \infty$.
Here the expectation is taken with respect to the joint distribution of $(\psi_y,\phi_t)$ as explained in Section \ref{Sec:formulation}.
This is a standard recursive identification method and constitutes a stochastic approximation sequence of the form:
\begin{equation}
    \bar\theta_i(t+1) = \bar\theta_i(t) + \alpha(t) \sbra{h_\theta(\bar\theta_i(t)) + M_\theta(t+1)},\ t \geq 0,
    \label{eq:sa-form}
\end{equation}
where $h_\theta(\bar\theta_i)=-\nabla\E{\|\epsilon(t)\|^2}$ is Lipschitz, and 
$M(t+1)=\nabla\E{\|\epsilon(t)\|^2}-\nabla \|\epsilon(t)\|^2$ is a Martingale difference
sequence.
This sequence converges almost surely to the equillibrium of
the differential equation 
\begin{equation}
    \dot {\bar{\theta_i}} = h_\theta(\bar\theta_i),\ t \geq 0.
    \label{eq:sa-ode}
\end{equation}
which can be shown to be a solution of \eqref{eq:sgd-min} 
with standard Lyapunov arguments. 
For more details the reader is referred to \cite{borkar2009stochastic,mavridis2023annealing}.
Moreover, notice that \eqref{eq:sgd-updates} is a vectorized representation of \eqref{eq:theta_recursion},
for $\gamma = \alpha(t)>0$. Therefore, under the PE condition \eqref{eq:PE-as} of Assumption \ref{as:PE}, 
and under the zero-mean noise assumption, it follows that $\bar\theta_i$ converges asymptotically to $\theta_i$
for all $i=1,\ldots,s$, i.e., the minimum of \eqref{eq:sgd-min} is achieved.

\subsection{Combined Mode and Dynamics Identification}
\label{sSec:combined}

Recall that the  
mode identification method is based on
the stochastic approximation updates 
\eqref{eq:oda_learning1}
that can be written with respect to the vectors $\xi_i(t) = [\hat\rho_i^\T(t) \sigma_i^\T(t)]^\T$ and a stepsize schedule $\beta(t)$ in the form:
\begin{equation}
    \xi_i(t+1) = \xi_i(t) + \beta(t) \sbra{h_\phi\pbra{\xi(t),\bar\theta(t)} + M_\phi(t+1)},\ t \geq 0,
    \label{eq:sa-phi}
\end{equation}
where $h_\phi$ is Lipschitz, $M_\phi(t)$ is a Martingale difference sequence and the dependence on $\bar\theta$ comes from the quantity $\hat p_i(t)$ in \eqref{eq:oda_learning2}
given \eqref{eq:theta_rule}.
At the same time, the recursive system identification technique to 
estimate $\bar\theta$ is a stochastic approximation 
sequence with a stepsize schedule $\alpha(t)$ of the form:
\begin{equation}
    \bar\theta_i(t+1) = \bar\theta_i(t) + \alpha(t) \sbra{h_\theta\pbra{\xi(t),\bar\theta(t)} + M_\theta(t+1)},\ t \geq 0,
    \label{eq:sa-theta}
\end{equation}
as given in \eqref{eq:sa-form}.
The dependence on $\xi$, comes through 
\eqref{eq:epsilon}, since $\xi$ defines $\hat\phi$,
which defines 
$\cbra{\Sigma_i}_{i=1}^{K}$ through \eqref{eq:Sigma-partition}, which defines 
$\cbra{S_i}_{i=1}^{\hat s}$ through the rule 
$\Sigma_j\in S_i$ if $\hat\theta_j(\bar\theta) = \bar\theta_i$.

Theorem \ref{thm:two_timescales} 
shows how the two recursive algorithms 
\eqref{eq:sa-phi} and \eqref{eq:sa-theta}
can be combined using the theory of two-timescale stochastic approximation
if $\nicefrac{\beta(t)}{\alpha(t)}\rightarrow 0$, i.e., 
when the estimation of the partition $\cbra{\Sigma_i}_{i=1}^K$ is updated at a slower 
rate than the updates of the parameters $\cbra{\bar\theta_i}_{i=1}^{\hat s}$.
\begin{theorem}
Consider the sequence $\cbra{\xi(t)}_{t\in\mathbb Z_+}$ generated using the updates 
\eqref{eq:sa-phi}, 
where $\xi_i(t) = [\hat\rho_i^\T(t) \sigma_i^\T(t)]^\T$, and $(\hat\rho_i, \sigma_i)$ are defined in \eqref{eq:oda_learning1}.
Consider the sequence $\cbra{\bar\theta(t)}_{t\in\mathbb Z_+}$
generated by the updates \eqref{eq:sa-theta}.
Let the stepsizes $(\alpha(t),\beta(t))$ of \eqref{eq:sa-theta} and \eqref{eq:sa-phi}, respectively, satisfy the conditions
$\sum_n \alpha(n) = \sum_n \beta(n) = \infty$, 
$\sum_n ( \alpha^2(n)+\beta^2(n) ) <\infty$, and 
$\nicefrac{\beta(n)}{\alpha(n)}\rightarrow 0$,
%
%
with the last condition implying that the iterations for $\cbra{\xi(t)}$
run on a slower timescale than those for $\cbra{\bar\theta(t)}$. 
If the equation 
\begin{equation}
\dot{\bar\theta}(t) = h_\theta(\xi,\bar\theta(t)),\ \bar\theta(0)=\bar\theta_0,
\label{eq:ode-theta}
\end{equation}
has an asymptotically stable equilibrium $\lambda(\xi)$ 
for fixed $\xi$ and some Lipschitz mapping $\lambda$, and the equation 
\begin{equation}
\dot \xi(t) = h_\phi(\xi(t),\lambda(\xi(t))),\ \xi(0)=\xi_0,
\label{eq:ode-xi}
\end{equation}
has an asymptotically stable equilibrium $\xi^*$, 
then, almost surely, the sequence $(\xi(t),\bar\theta(t))$
generated by \eqref{eq:sa-phi}, \eqref{eq:sa-theta}, converges to $(\xi^*,\lambda(\xi^*))$.
\label{thm:two_timescales}
\end{theorem}
\begin{proof}
    It follows directly from Theorem 2, Ch. 6 of \cite{borkar2009stochastic}.
\end{proof}

\begin{corollary}
    Condition \eqref{eq:ode-theta} of Theorem \ref{thm:two_timescales} is satisfied by the definition of $h_\theta$ in \eqref{eq:sa-ode}. Therefore,
    \eqref{eq:ode-xi} implies the convergence of $\hat\phi$ through \eqref{eq:oda_learning2}, and of the partition $\cbra{\Sigma_i}$ through 
    \eqref{eq:Sigma-partition}.
\end{corollary}

Notice that the condition 
$\nicefrac{\beta(t)}{\alpha(t)}\rightarrow 0$ is of great importance.
Intuitively, the stochastic approximation algorithm 
\eqref{eq:sa-phi}, \eqref{eq:sa-theta} 
consists of two components running in different timescales, where
the slow component is viewed as quasi-static 
when analyzing the behavior of the fast transient.
In practice, the condition $\nicefrac{\beta(t)}{\alpha(t)}\rightarrow 0$
is satisfied by stepsizes of the form 
$(\alpha(t),\beta(t))=(\nicefrac 1 t, \nicefrac{1}{(1+t \log t)})$, or
$(\alpha(t),\beta(t))=(\nicefrac{1}{t^{\nicefrac 2 3}}, \nicefrac{1}{t})$.
Another way of achieving the two-timescale effect is to 
run the iterations for the slow component 
with stepsizes $\cbra{\alpha_{t(k)}}$, 
where $t(k)$ is a subsequence of $t$ that becomes increasingly rare
(i.e. $t(k+1)-t(k)\rightarrow\infty$), 
while keeping its values constant between these instants.
A good policy is to combine both approaches and 
update the slow component with slower stepsize schedule $\beta(t)$ 
along a subsequence keeping its values constant in between
(e.g., \cite{mavridis2023annealing,borkar2009stochastic}).

\section{General Switched System Identification}
\label{Sec:SwitchedID}

In Sections \ref{Sec:formulation}, \ref{Sec:ODA}, and 
\ref{Sec:PWAID} we have developed a real-time idenitification method for 
PWA systems. 
However, neither the proposed methodology, nor the algorithmic implementation 
of Alg. \ref{alg:ODA} are constrained to PWA systems. 
Thus the proposed approach can, in principle, be applied to more general 
switching and hybrid systems. 
However, issues may arise with respect to 
the identifiability conditions, the mode-switching estimation error, and
the possibly non-linear local system identification error.
In this section, we discuss the applicability of the proposed approach in 
different cases often encountered in switching control systems, namely switched linear systems with non-polyhedral partition, and switched non-linear systems with polyhedral partition.

\subsection{Switched linear systems with non-polyhedral partition.}

In the case of linear local dynamics, 
the recursive identification method discussed in Section \ref{sSec:gradient-descent} 
remains unchanged, and the same convergence results hold.
However, 
if the regions $S_i$ of the mode switching partition $\cbra{S_i}_{i=1}^s$ are non-polyhedral,
they cannot be perfectly approximated by a finite union of polyhedral regions $\cbra{\Sigma_i}_{i=1}^K$.
%
%
It is worth pointing out that from the convergence results of the online deterministic annealing 
algorithm \cite{mavridis2023multi}, it follows that the partition error can be arbitrarily 
small in the limit $K\rightarrow \infty$ (which is the case when $\lambda\rightarrow 0$).
Albeit a nice analytical result, in practice there will always be non-zero error in the estimation 
of the partition $\cbra{S_i}_{i=1}^s$.
We hereby discuss two ways to deal with this problem. 
The first is to assume the existence of a non-linear transformation that maps each $S_i$
to a polyhedral region $\bar S_i$, and proceed with Alg. \ref{alg:ODA}. 
General-purpose learning machines, 
such as artificial neural networks can be incorporated in this process.
Further assumptions and analysis is required for this method, which is beyond the scope of this paper.
The second refers to mitigating the jumping effect of the identified system to decrease the 
closed-loop error that naturally occurs due to imperfect mode switching. 
To this end, recall that, given an observation $\phi_t$ the dynamics of the identified model
are given according to \eqref{eq:pwa-theta} by:
\begin{equation}
    \hat \psi_t = [\phi_t^\T \otimes I_{m}] \bar\theta_i, \text{ if } \phi_t\in \Sigma_j 
    \text{ and } \hat\theta_j(\bar\theta) = \bar\theta_i.
\end{equation}
To mitigate the jumping behavior one can make use of the association probabilities 
\begin{equation}
p(\phi_i|\phi_t) = \frac{e^{-\frac{1-\lambda}{\lambda}d_\rho(\phi_t,\phi_i)}}
			{\sum_j e^{-\frac{1-\lambda}{\lambda}d_\rho(\phi_t,\phi_i)}},
\end{equation}
to instead construct the weighted dynamics:
\begin{equation}
    \hat \psi_t = \sum_{i=1}^K p^*(\phi_i|\phi_t) [\phi_t^\T \otimes I_{m}] \hat\theta_i.
\end{equation}
This jump-mitigation method has been used in the literature 
to preserve smoothness of the closed-loop dynamics and is particularly useful when 
hybrid system identification is used for non-linear function approximation, i.e., 
when the original system is not hybrid but is to be approximated by a hybrid system 
with simpler local dynamics.

\subsection{Switched non-linear systems with polyhedral partition.}

In this case, often referred to as 
piece-wise non-linear hybrid systems \cite{lauer2008switched},
the mode switching partition $\cbra{S_i}_{i=1}^s$ is polyhedral,
and can be perfectly approximated by a finite union of polyhedral regions 
$\cbra{\Sigma_i}_{i=1}^K$.
For the identification of the non-linear local dynamics, 
the recursive identification method discussed in Section \ref{sSec:gradient-descent} 
needs to be modified.
In particular the recursive updates:
\begin{equation}
\begin{aligned}
\bar\theta_i(t+1) &= \bar\theta_i(t) - \alpha(t)  \pbra{\nabla_{\bar\theta_i} \epsilon(t)} \epsilon(t),
\end{aligned}
\end{equation}
given in \eqref{eq:sgd-updates} of the same stochastic 
gradient descent structure are used, with the error term in this case given by 
\begin{equation}
    \epsilon(t) =  \sum_i \mathds{1}_{\sbra{\phi_t\in S_i}} \hat f(\phi_t,\bar\theta_i) - \psi_t, 
\end{equation}
where the functions $\hat f(\phi_t,\bar\theta_i)$ are local parametric models of known form, 
differentiable with respect to the parameters $\theta_i$.
General-purpose learning machines, 
such as artificial neural networks can be used.
Notice that the identification updates remain stochastic approximation updates of the same form, 
which means that the convergence results of Theorem \ref{thm:two_timescales} continue to hold.



\section{Experimental Results}
\label{Sec:Results}

We illustrate the properties and evaluate the performance 
of the proposed algorithm 
%
%
in multiple PWA systems, both in PWARX and state-space form.

\subsection{Benchmark PWARX System}

The benchmark PWARX system, adopted from \cite{ferrari2003clustering},  is given in the input--output representation of \eqref{eq:exp1-system}:
\begin{equation}
    \begin{aligned}
        y_t = \begin{cases}
            \theta_1^\T \phi_t + e_t,\quad \text{if } r_t\in P_1 \\
            \theta_2^\T \phi_t + e_t,\quad \text{if } r_t\in P_2 \\
            \theta_3^\T \phi_t + e_t,\quad \text{if } r_t\in P_3 \\
              \end{cases},
    \end{aligned}
    \label{eq:exp1-system}
\end{equation}
where $y_t\in\mathbb R^1$, $r_t\in P=[-4,4]$, $\phi_t=[r_t\ 1]^\T$,
$(P_1,P_2,P_3) = ( [-4,-1], (-1,2), [2,4] )$, and 
$(\theta_1,\theta_2,\theta_3) = ( [1,2]^\T, [-1,0]^\T, [1,2]^\T )$.
The simplicity of this example enables the graphical representation of the mode-switching partition
and the convergence of the model parameters.
At the same time, \eqref{eq:exp1-system} presents a jump at $r_t=2$,
and same dynamics for different regions of the input space, 
i.e., $\theta_1=\theta_3$ while $P_1\neq P_3$.
It can thus be written in the form:
\begin{equation}
    \begin{aligned}
        y_t = \begin{cases}
            \theta_2^\T \phi_t + e_t,\quad \text{if } \phi_t\in S_2 \\
            \theta_1^\T \phi_t + e_t,\quad \text{otherwise}
              \end{cases},
    \end{aligned}
    \label{eq:exp1-system-modes}
\end{equation}
where $S_2=\cbra{\phi=[r\ 1]^\T:r\in P_2}$.
%
A total of $N=150$ observations under Gaussian noise ($e_t\sim \mathcal N(0,0.2)$) 
are accessible sequentially.

%
%

Algorithm \ref{alg:ODA} is applied to the observations for $T=900$ iterations.
%
%
The temperature parameters used for the online deterministic annealing algorithm are 
$(\lambda_{\max},\lambda_{\min},\gamma)=(0.99,0.2,0.8)$, and the stepsizes 
$(\alpha(t),\beta(t))=(\nicefrac{1}{(1+0.01t)}, \nicefrac{1}{(1+0.9t \log t)})$.
In addition, $\delta=0.1$, $\epsilon_n = 1.0$, and $\epsilon_s = 2.0$.
At first ($\lambda = \lambda_{\max}$), the algorithm keeps in memory only 
one codevector $\hat\phi_1$ and one model parameter vector $\bar\theta_1$, 
essentially assuming that the system has constant dynamics 
in the entire domain, i.e., $\hat S_1=\Sigma_1=\cbra{\phi=[r\ 1]^\T:r\in P}$.
As new input--output pairs are observed, the estimated parameter $\bar\theta_1$ gets 
updated by the iterations \eqref{eq:sgd-updates}.
We have assumed $\bar\theta_1(0)=[1,1]^\T$.

At the same time, the estimate of $\bar\theta_1$ are used to update the location 
of the codevector towards the mean of the observation domain as shown in \eqref{eq:mu_star}. 
%
%
As $\lambda$ is reduced, the bifurcation phenomenon described in Section \ref{sSec:bifurcation} takes place, and, after reaching a critical value, the single
codevector splits into two duplicates. 
Now the algorithm assumes that there are two modes in the system and estimates the optimal
model parameters $\cbra{\bar\theta_1,\bar\theta_2}$ and partition $\cbra{\Sigma_1,\Sigma_2}$
(through the location of the codevectors $\cbra{\hat\phi_1,\hat\phi_2}$).
This process continues until a desired termination criterion is reached. In this case it is the minimum temperature parameter $\lambda_{\min}$ that reflects to a 
potential time and computational constraint of the system.
The bifurcation phenomenon is illustrated in Fig. \ref{fig:bifurcation} where the 
locations of the codevectors $\cbra{\hat\phi_i}$, $\hat\phi_i\in P=[-4,4]$ generated by 
Alg. \ref{alg:ODA} are shown.
The algorithm progressively constructs a total of $K=5$ effective codevectors.
The number of modes is estimated with the process explained in 
Section \ref{sSec:modes}.
Two modes are estimated
with
$\bar\theta=\cbra{\bar\theta_1,\bar\theta_2}$. 
The association of each effective codevector with each identified mode according to the rule \eqref{eq:theta_rule} is shown in Fig. \ref{fig:bifurcation}.
\begin{figure}[t]
\centering
\includegraphics[trim=20 20 0 0,clip,width=0.45\textwidth]{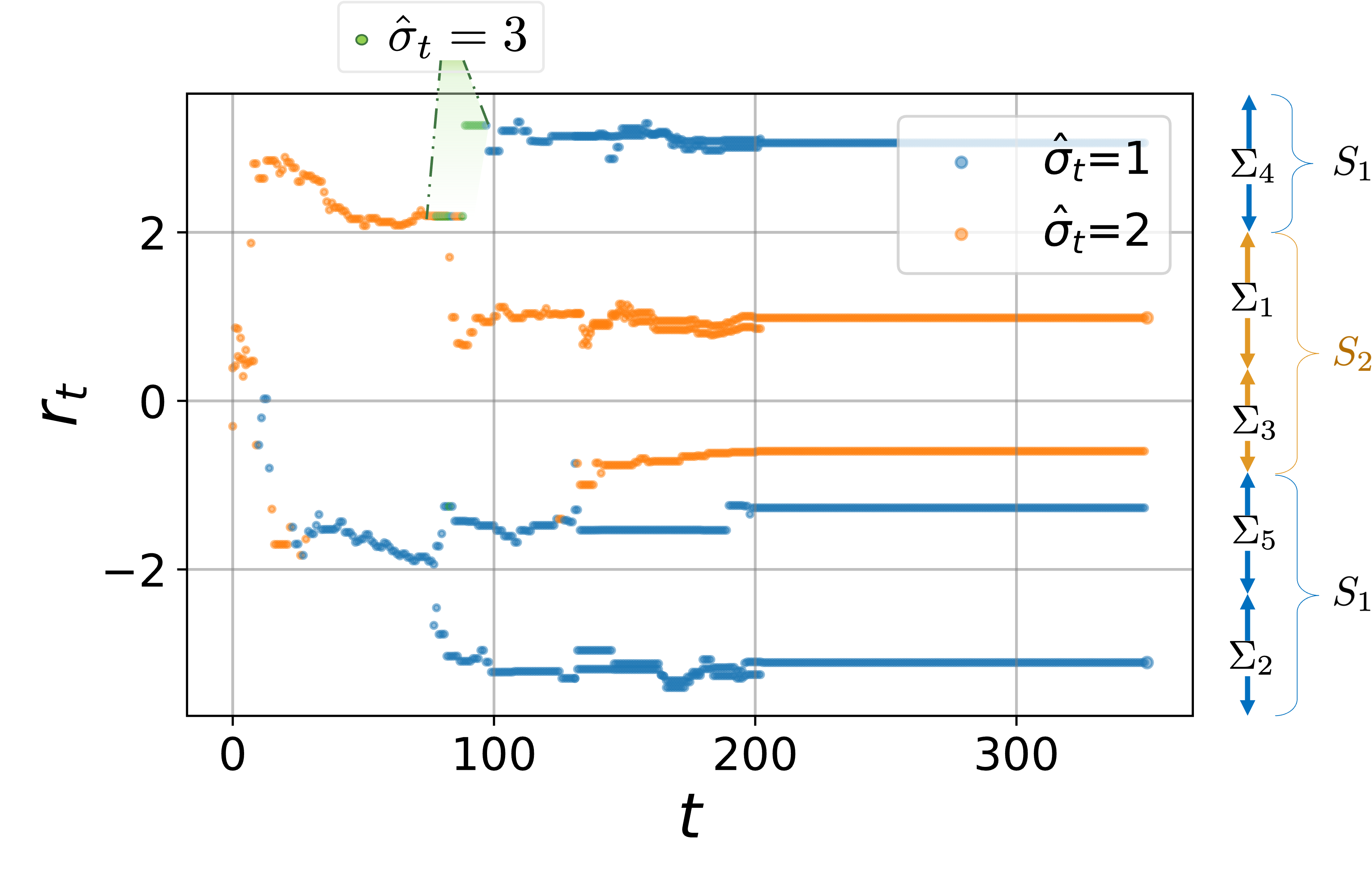}
\caption{Evolution of the codevectors $\cbra{\hat\phi_i}$ generated by 
Alg. \ref{alg:ODA} for system \eqref{eq:exp1-system-modes} illustrating the bifurcation phenomenon described in Section \ref{sSec:bifurcation}. 
The association of each effective codevector
with each identified mode according to the rule \eqref{eq:theta_rule} is also shown. Notice that a third mode is detected and quickly discarded as explained in Section \ref{sSec:modes}}.
\label{fig:bifurcation}
\end{figure}
\begin{figure}[t]
\centering
%
\includegraphics[trim=10 10 0 0,clip,width=0.45\textwidth]{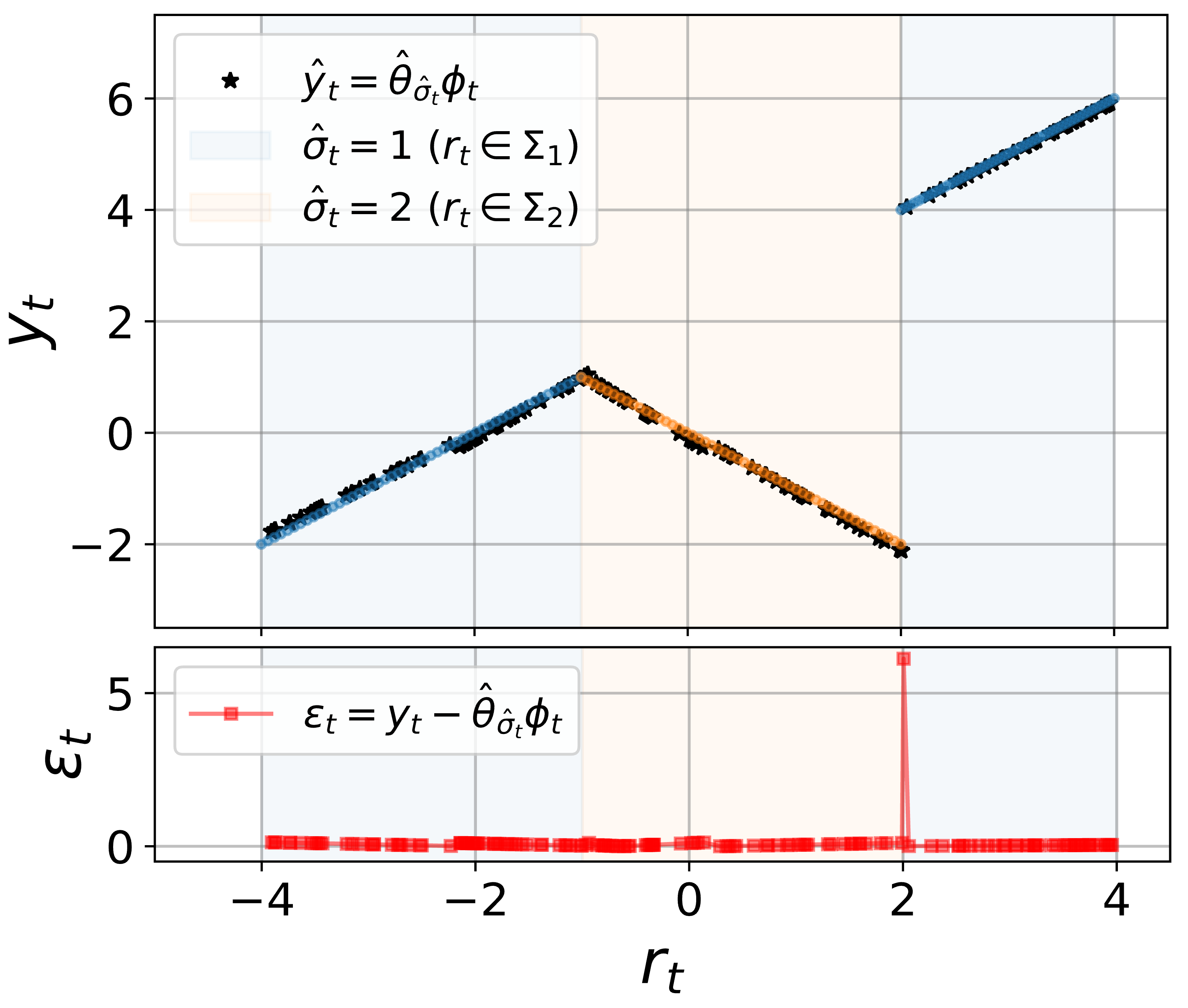}
%
\caption{Identified modes, predicted output and identification error with 
respect to the true model \eqref{eq:exp1-system-modes}. 
}
\label{fig:prediction}
\vspace{-1em}
\end{figure}
The final estimated partition, the output of the estimated model, and its error with 
respect to the true model without noise are shown in Fig. \ref{fig:prediction}. 
A single 
misclassification instance of the mode at the boundary of the true partition of the 
input--output domain is observed.
This mode switching error can be avoided by allowing $\lambda$ to go lower, which results in a larger number $K$ of effective codevectors and is indicative of the performance/complexity trade-off of the algorithm.
%
%
%
Finally, regarding the effect of the noise variance on the identification error,
for $e_t\sim \mathcal N(0,0.2)$, 
$e_t\sim \mathcal N(0,0.5)$ and $e_t\sim \mathcal N(0,0.7)$, the root-mean-square deviation across the observed samples was computed as $e_{RMS} = 0.504$, $e_{RMS} = 0.642$, and $e_{RMS} = 0.689$, respectively.

%
%


\subsection{Comparison with existing methods}

Compared to the clustering-based method in \cite{ferrari2003clustering}, the proposed algorithm applied to system \eqref{eq:exp1-system-modes} shows similar performance while maintaining several advantages. First, the number of modes is not assumed to be known a priori.
Second, the proposed identification method can operate in real-time, i.e., using one forward pass of online observations as opposed to iterating multiple times through a dataset acquired offline. 
Finally, the progressive nature of the algorithm allows for the exploitation of the performance/complexity trade-off in applications where communication or computational resources are limited.

The same advantages can be observed against more recent methods as well.
Consider the following system: 
%
\begin{equation}
    \begin{aligned}
        y_t = \begin{cases}
            \theta_1^\T \phi_t + e_t,\quad \text{if } \phi_t\in S_1 \\
            \theta_2^\T \phi_t + e_t,\quad \text{otherwise}
              \end{cases},
    \end{aligned}
    \label{eq:expB-system-modes}
\end{equation}
where $u_t, y_t\in\mathbb R$, $\phi_t = [y_{t-1}, y_{t-2}, u_{t-1}, u_{t-2}, 1]\in\mathbb R^5$, 
$\theta_1=[0.1,0.5,-0.4,0.3,0]^\T$,
$\theta_2=[0.2,0.4,0.1,0.4,0]^\T$,
and 
$S_1=\cbra{\phi\in\mathbb R^5: [1,0.5,-0.3,2,0.2]^\T\phi \geq 0}$.
Also define a ``best fit rate'' objective $b_f$ as:
%
\begin{equation}
    b_f = 1-\sqrt{\frac{\sum_t\|y_t-\hat y_t\|^2}{\sum_t\|y_t- \bar y\|^2}},
\end{equation}
where $\bar y$ represents the numerical mean value of $\cbra{y_t}_{t\geq0}$.
In this system, simulated for $t\in[0,T],\ T=10000$, the method 
proposed in \cite{bemporad2018fitting} achieves $b_f^1 = 0.6568$ in $\tau_1 = 52.28$ seconds \cite{bemporad2018fitting}. The proposed method achieves $b_f^2 = 0.7792$, constructing $K=6$ codevectors $\hat s=2$ modes with parameter vectors 
$\bar \theta_1=[0.07,0.48,-0.40,0.29,0.00]^\T$ and
$\bar \theta_2=[0.13,0.43,-0.03,0.58,0.01]^\T$.
In the same desktop machine, the forward loop of the system 
including the identification computation overhead of the proposed algorithm, lasted a total of $\tau_2 = 16.13$ seconds. This allows real-time operation for systems of the form \eqref{eq:expB-system-modes} sampled at frequency $f_s=620$ Hz or lower, i.e., with sampling period $T_s=0.0016$ sec or higher. 
Actual and predicted trajectories for system \eqref{eq:expB-system-modes} for the first $T=100$ timesteps are depicted in Fig. \ref{fig:prediction_bemporad}. Mode changes are depicted as background color.

\begin{figure}[h]
\centering
\includegraphics[trim=0 46 0 0,clip,width=0.48\textwidth]{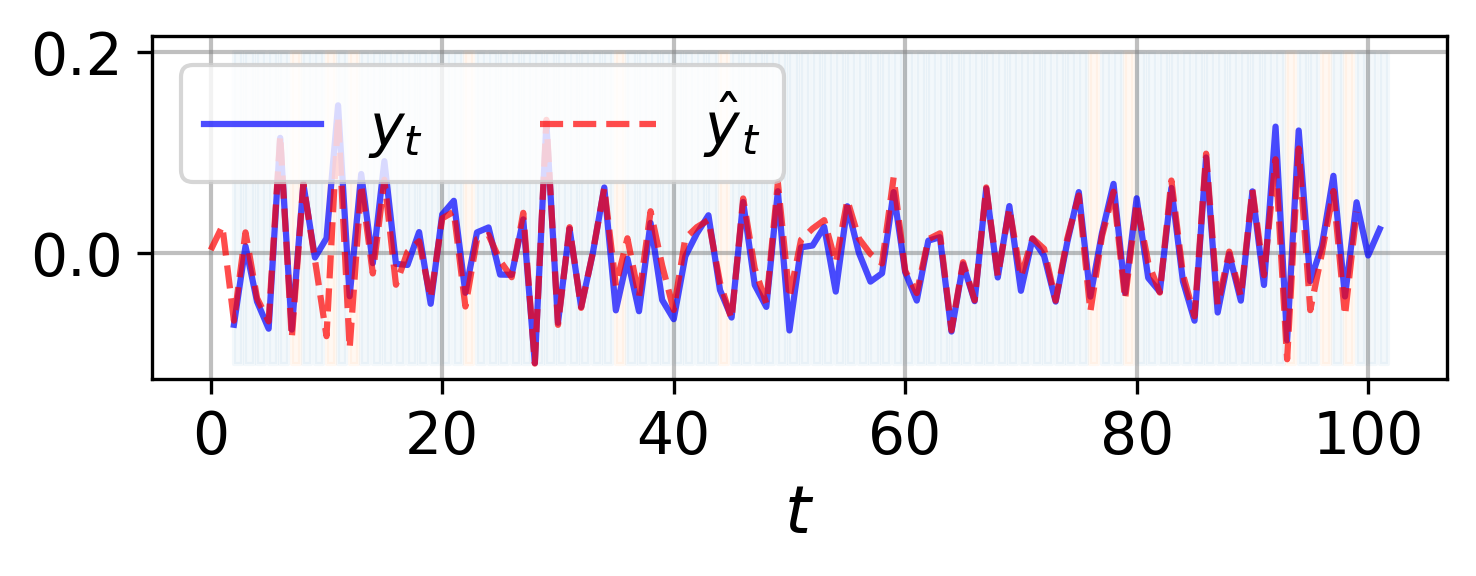}
\includegraphics[trim=0 10 0 0,clip,width=0.48\textwidth]{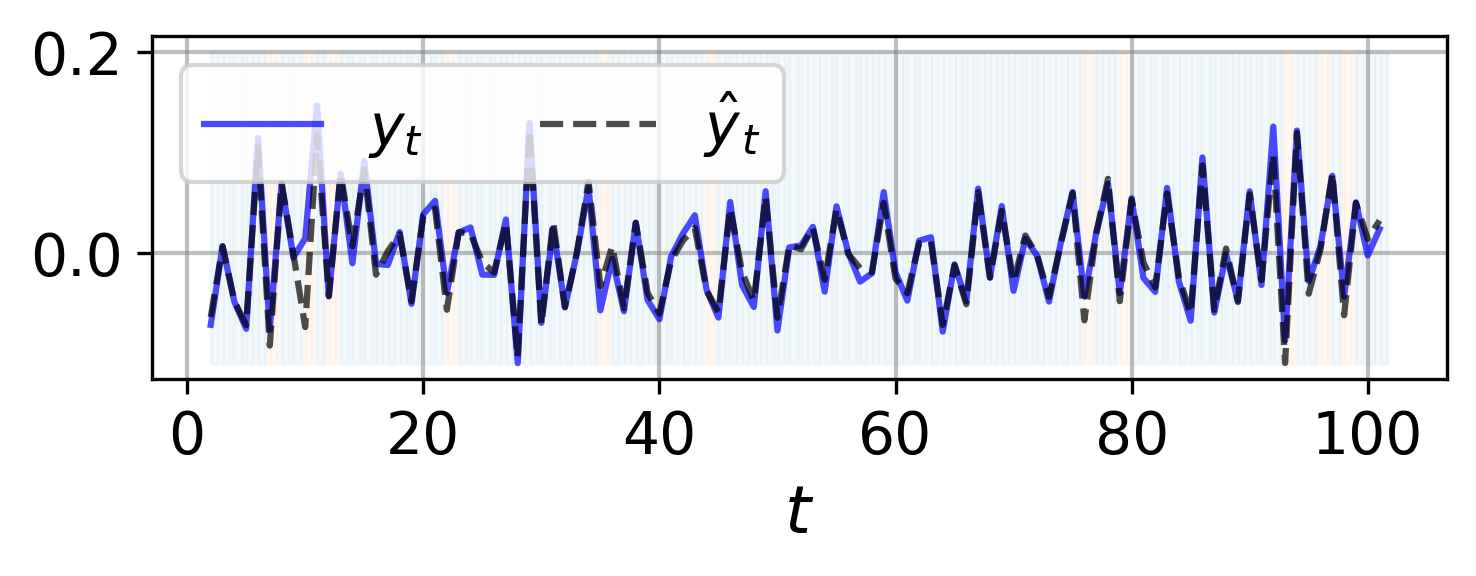}
\caption{Actual and predicted trajectories for system \eqref{eq:expB-system-modes} using the proposed method (black) and the method in \cite{bemporad2018fitting} (red).
}
\label{fig:prediction_bemporad}
\end{figure}

\subsection{State-Space PWA System}

As a simple state-space PWA example, consider following linearized PWA dynamics in the state-space domain:
\begin{equation}
    \begin{aligned}
        \begin{cases}
            x_{t+1} &= (I_{2} + \textrm{d}t 
            \begin{bmatrix} 0 & 1 \\ 
                            0 & 0 
            \end{bmatrix}) x_t + \textrm{d}t
            \begin{bmatrix} 0 \\ 1 \end{bmatrix} u_t + e_t,\ |u_t|> 1 \\
            x_{t+1} &= (I_{2} + \textrm{d}t 
            \begin{bmatrix} 0 & 1 \\ 
                            0 & -1 
            \end{bmatrix}) x_t + \textrm{d}t
            \begin{bmatrix} 0 \\ 0 \end{bmatrix} u_t + e_t,\ |u_t|\leq 1
        \end{cases},
    \end{aligned}
    \label{eq:exp2-system}
\end{equation}
where $x_t\in\mathbb R^2$, $u_t\in \mathbb R$, and $e_t\sim \mathcal N(0,0.5)$.
System \eqref{eq:exp2-system} has two modes ($s=2$)  
and the switching signal is defined by the polyhedral regions 
$R_1=\cbra{[x^\T|u^T]^\T\in \mathbb R^3:u<-1}$,
$R_2=\cbra{[x^\T|u^T]^\T\in \mathbb R^3:-1<u<1}$, and
$R_3=\cbra{[x^\T|u^T]^\T\in \mathbb R^3: 1<u}$
with $S_1 = R_1\bigcup R_3$ and $S_2 = R_2$.
The dynamics of \eqref{eq:exp2-system} consist of a controllable double integrator
when the input is of sufficient magnitude, and a stable autonomous system, otherwise.
%
%
In this example, the linear system of the second mode ($s=2$) is not minimal, 
and its identification relies on the mode switching behavior of the system, 
as explained in Section \ref{sSec:pwa-identifiability}.
%
%
%
To preserve the PE conditions of Assumption \ref{as:PE}, 
the input signal is chosen as $u_t = 2\cos(2\pi t* \textrm{d}t)$, $t\in\mathbb Z_+$,
and the noise term $w_t$ is a zero-mean Gaussian random variable with $\sigma^2=0.1$.
%
%

%
%

The system is allowed to run for $T=3s$ (seconds), with $\textrm{d}t=0.01$, i.e.,
a total of $N=300$ observations are acquired online, 
based on which, the proposed method identifies the switched system in real time.
The temperature parameters used for the online deterministic annealing algorithm are 
$(\lambda_{\max},\lambda_{\min},\gamma)=(0.99,0.1,0.8)$, 
$\delta=0.1$, $\epsilon_n = 0.5$, and $\epsilon_s = 0.01$,
and
$(\alpha(t),\beta(t))=(\nicefrac{1}{1+0.01t}, \nicefrac{1}{1+0.9t \log t})$.
%
%
The estimated parameter $\hat\theta_1$ gets 
updated by the iterations \eqref{eq:sgd-updates}.
We have assumed $\hat\theta_1(0)=[0,1,1,0,1,1]^\T$.
%
%
%
A total of $K=4$ effective codevectors and $\hat s=2$ modes are estimated.
%

The identification error and the estimated mode switching error are shown in Fig. \ref{fig:prediction2} in comparison with the true mode switching behavior of the system. 
More specifically, the algorithm identifies a total of $\hat s=2$ modes with 
$S_1=\Sigma_3\bigcup \Sigma_4$ and
$S_2=\Sigma_1\bigcup \Sigma_2$.
In Figure \ref{fig:convergence2},
the convergence of the parameters $\cbra{\bar\theta_i}$ 
of each of the $\hat s=2$ local models detected to the actual $\cbra{\theta_{i}}_{i=1}^2$
observed are shown.
Parameter values that do not appear at $t=0$ indicate that they belong to modes identified through the bifurcation phenomenon 
after a certain critical temperature value.

%
%

%
\begin{figure}[t]
\centering
\includegraphics[trim=0 15 0 0,clip,width=0.48\textwidth]{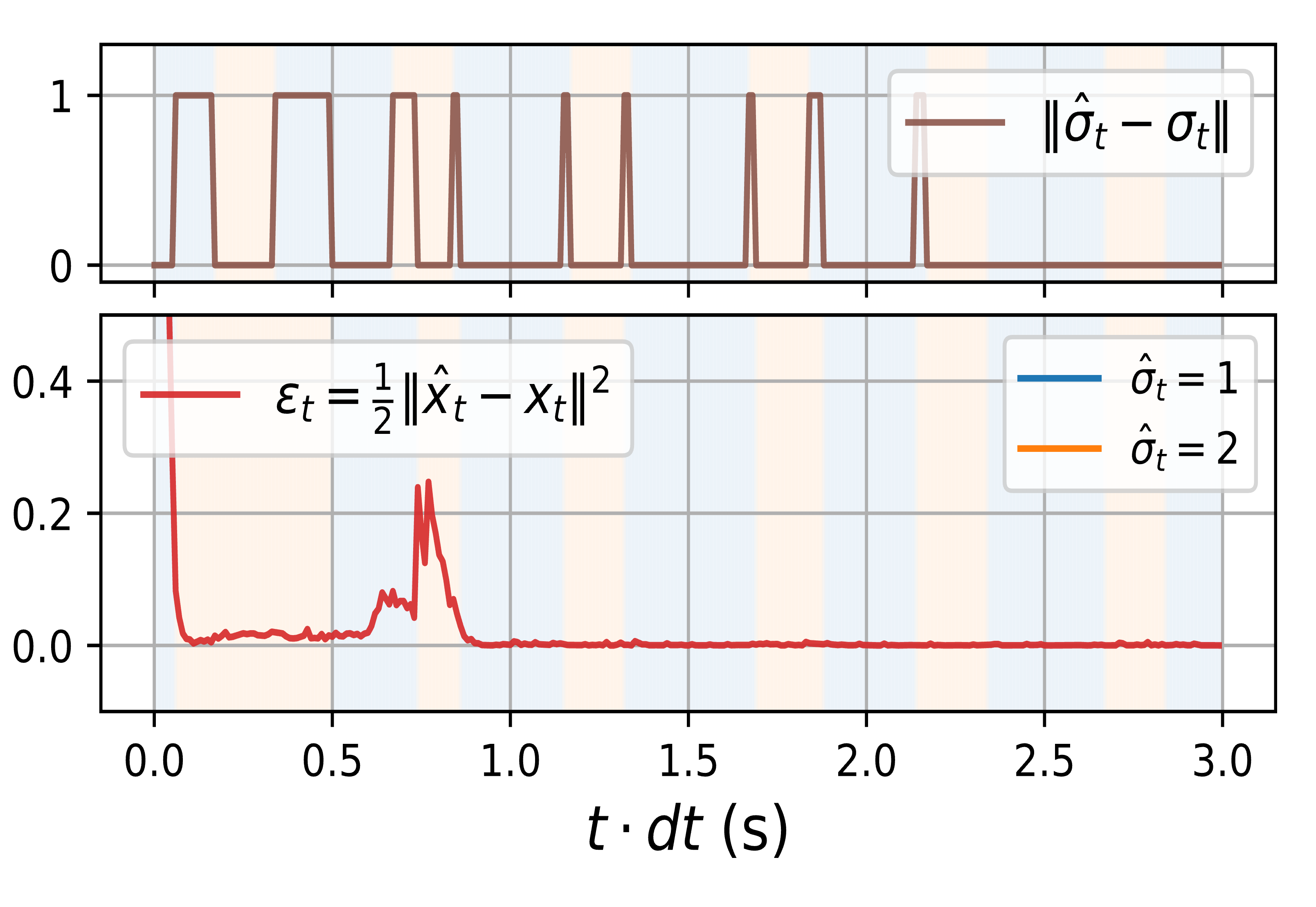}
%
\caption{Identification error over time for system \eqref{eq:exp2-system}. 
The estimated modes are also compared against the original modes.}
\label{fig:prediction2}
\vspace{-1em}
\end{figure}
\begin{figure}[t]
\centering
%
\includegraphics[trim=0 20 0 0,clip,width=0.45\textwidth]{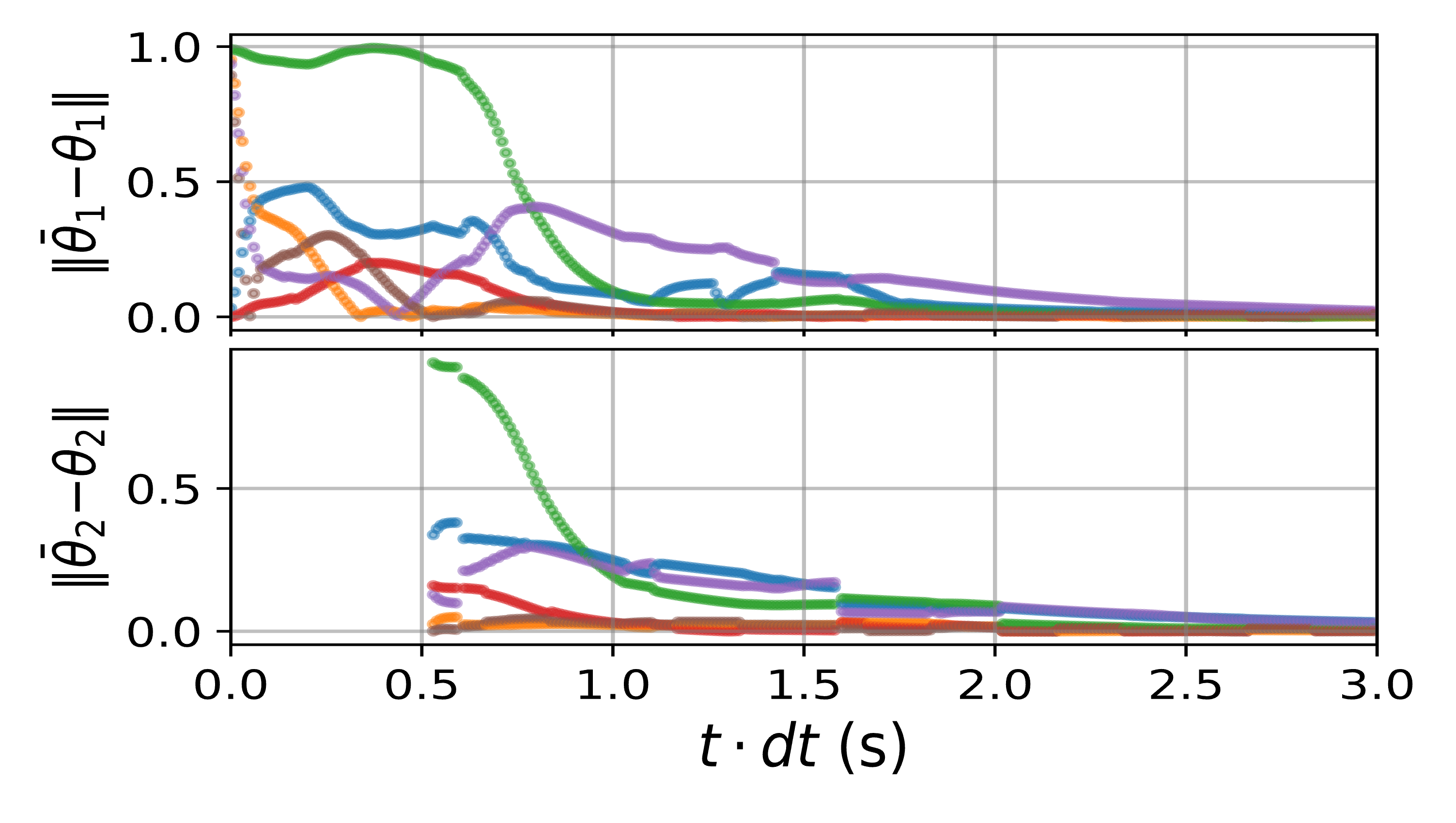}
%
\caption{Convergence of the parameters $\cbra{\bar\theta_i}_{i=1}^2$ to the true values of \eqref{eq:exp2-system}.
Parameter values that do not appear at $t=0$ indicate that belong to 
modes identified through the bifurcation phenomenon 
described in Section \ref{sSec:bifurcation}.}
\label{fig:convergence2}
\vspace{-1em}
\end{figure}

\subsection{Higher-dimensional State-Space PWA System}

A higher dimensional example is given by the following linearized PWA dynamics in the state-space domain:
\begin{equation}
\begin{aligned}
x_{t+1} = (I_{2} + \textrm{d}t 
        \begin{bmatrix}
        \begin{bmatrix} 0 & 1 & 0\\ 
                      0 & 0 & 1   
        \end{bmatrix}\\
         a_i 
        \end{bmatrix}
        x_t 
+ \textrm{d}t B_i u_t + e_t,\ 
    \begin{bmatrix} x_t\\ u_t  
        \end{bmatrix}\in S_i \\
\end{aligned}
\label{eq:exp3-system}
\end{equation}
where $x_t\in\mathbb R^3$, $u_t\in \mathbb R$, $e_t\sim \mathcal N(0,0.2)$, $i\in \cbra{1,2,3}$, 
$a_1=[-1,-2,-1]^\T$, $B_1=[0,0,1]^\T$,
$a_2=[-1,-1,-1]^\T$, $B_2=[0,1,0]^\T$,
$a_3=[-1,-2,-2]^\T$, $B_3=[0,0,1]^\T$,
$S_1 = \cbra{r\in\mathbb R^4: [0,1,0,1]^\T r \geq 1}$
$S_2 = \cbra{r\in\mathbb R^4: -1 < [0,1,0,1]^\T r < 1}$
$S_3 = \cbra{r\in\mathbb R^4: [0,1,0,1]^\T r \leq -1}$.
The online deterministic annealing algorithm parameters are chosen as
$(\lambda_{\max},\lambda_{\min},\gamma)=(0.99,0.1,0.8)$, 
$\delta=0.1$, $\epsilon_n = 0.5$, and $\epsilon_s = 0.01$,
and
$(\alpha(t),\beta(t))=(\nicefrac{1}{1+0.01t}, \nicefrac{1}{1+0.9t \log t})$.
We have also assumed $\hat\theta_1(0)=[0,0.1,0,0,0,0.1,-0.5,-0.5,-0.5,0.1,0.1,0.1]^\T$.
A total of $K=4$ effective codevectors and $\hat s=3$ modes are estimated.
The identification error and the estimated mode switching behavior are shown in Fig. \ref{fig:prediction3} in comparison with the true mode switching behavior of the system.

\begin{figure}[h]
\centering
\includegraphics[trim=0 15 0 0,clip,width=0.48\textwidth]{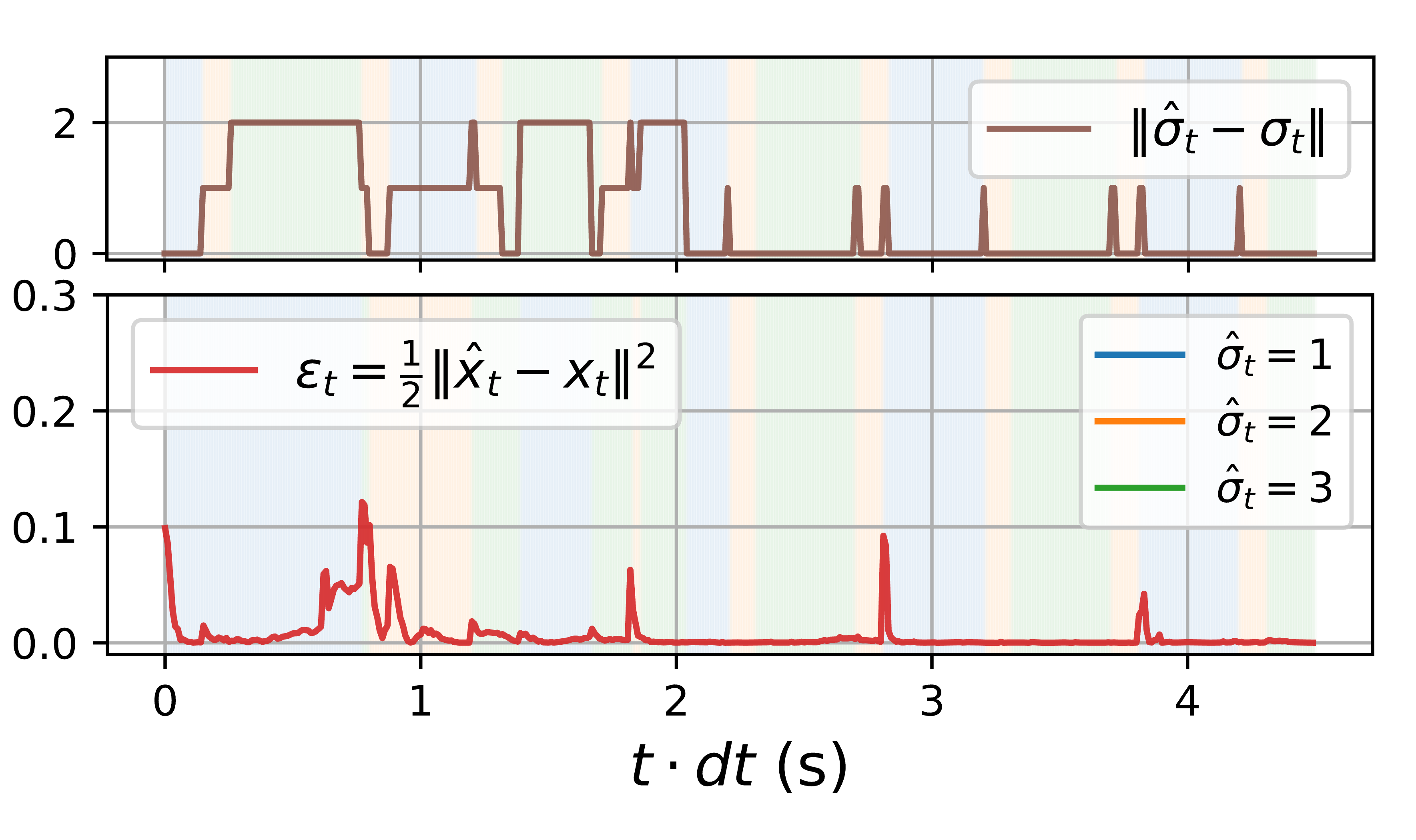}
\caption{Identification error over time for system \eqref{eq:exp3-system}. 
The estimated modes are also compared against the original modes.}
\label{fig:prediction3}
\vspace{-1em}
\end{figure}
%

\section{Conclusion}
\label{Sec:Conclusion}

A real-time system identification scheme is proposed,
appropriate for online identification
of both the modes and the subsystems of a discrete-time switched system.
The proposed method is computationally 
efficient compared to standard algebraic, 
mixed-integer programming, and offline clustering-based methods,
%
and provides real-time control over the 
performance-complexity trade-off.
Future directions will focus
on the development of an adaptive annealing schedule with respect to time-dependent changes and extensions to 
identification of both discrete- and continuous-time
partially observable piece-wise affine models in the state-space domain 
using real-time observations.
%


\bibliographystyle{IEEEtran} %
\bibliography{bib_learning.bib,bib_mavridis.bib,bib-switched.bib}


\appendices

\section{Proof of Theorem \ref{thm:identification}.} 
\label{App:identification}

We construct the system 
\begin{equation}
    \hat x_{t+1} = \hat A x_t + \hat B u_t,\quad t\in\mathbb{Z}_{+},
    \label{eq:linearhat}
\end{equation}
where $\hat A\in \mathbb{R}^{n\times n}$, and $\hat B\in \mathbb{R}^{n\times p}$.
Subtracting \eqref{eq:linear} from \eqref{eq:linearhat}, we get:
\begin{equation}
    e_{t+1} = \bar \Theta r_t,\quad t\in\mathbb{Z}_{+},
\end{equation}
where 
$e_t = \hat x_t - x_t \in\mathbb R^{n}$ is the observation error,
$r_t = [x_t^\T | u_t^\T]^\T \in\mathbb R^{n+p}$ is the augmented state-input vector as defined in \eqref{eq:r-vector}, and
$\bar \Theta = [(\hat A - A) | (\hat B - B)]$ is an augmented matrix of the system parameters of size $n\times(n+p)$.
Then \eqref{eq:theta_recursion} is equivalent to:
\begin{equation}
    \bar \Theta_{t+1} = \bar \Theta_t - \gamma e_{t+1} r_t^\T,\quad t\geq 0.
    \label{eq:theta_recursion_proof}
\end{equation}
Notice that \eqref{eq:theta_recursion_proof} can be written in the form
of a linear time-varying dynamical system:
\begin{equation}
    \bar \Theta_{t+1} = \bar \Theta_t ( I_{n+p} - \gamma r_t r_t^\T ),~ t\geq 0.
    \label{eq:ltv}
\end{equation}
By vectorizing $\bar \Theta_t$ such that $\bar \theta_t = \text{vec}(\bar \Theta_t)$, \eqref{eq:ltv} becomes:
\begin{equation}
\begin{aligned}
    \bar \theta_{t+1} &=  ( I_{n(n+p)} - \gamma \psi_t \psi_t^\T ) \bar \theta_t = \Xi_t \bar \theta_t,~ t\geq 0,
\end{aligned}
    \label{eq:ltv-vec}
\end{equation}
where $\otimes$ denotes the Kronecker product, and $\psi_t = [r_t^\T \otimes I_n]^\T$ is a $n(n+p)\times n$ matrix.
We will show that \eqref{eq:ltv-vec} is exponentially stable in the large (Definition 1, \cite{jenkins2018convergence})
as long as \eqref{eq:PE} is satisfied.
Consider the Lyapunov function candidate $V(t,\bar \theta) = \bar \theta_t^\T \bar \theta_t$.
It is obvious that there exist $k_1,k_2>0$ such that 
$k_1 \|\bar \theta\|^2\leq V(t,\bar \theta) \leq k_2 \|\bar \theta\|^2$.
Notice that 
$V(t+1,\bar \theta_{t+1})-V(t,\bar \theta_{t}) = \bar \theta_t^\T \Xi_t^\T \Xi^t \bar \theta_t$.
As a result, by summing the differences for $T$ timesteps, we get:
\begin{equation}
\begin{aligned}
V(t+&T+1,\bar \theta_{t+T+1})-V(t,\bar \theta_{t}) = \\
&= \sum_{\tau=t}^{t+T} V(\tau+1,\bar \theta_{\tau+1})-V(\tau,\bar \theta_{\tau}) \\
&= \sum_{\tau=t}^{t+T} \bar \theta_\tau^\T \pbra{ \Xi_\tau^\T \Xi_\tau - I_{n(n+p)}} \bar \theta_\tau \\
&= \bar \theta_t^\T \sbra{ \sum_{\tau=t}^{t+T} \Phi(\tau;t)^\T \pbra{ \Xi_\tau^\T \Xi_\tau - I_{n(n+p)}} \Phi(\tau;t) } \bar \theta_\tau \\
&\leq - \alpha_1 \bar \theta_t^\T I_{n(n+p)} \bar \theta_t = - \alpha_1 V(t,\bar \theta_{t}),
\end{aligned}
\label{eq:lyapunov-theta}
\end{equation}
for some $0<\alpha_1<1$. 
Here $\Phi(\tau;t)=\Xi_{t}\Xi_{t+1}\ldots \Xi_{\tau-1}$ is the transition matrix of \eqref{eq:ltv-vec}, and 
the inequality follows from condition \eqref{eq:PE}.
Notice that the first inequality in \eqref{eq:PE} is equivalent to 
$\alpha I_{n+p} \preceq \sum_{\tau=t}^{t+T} r_\tau^\T r_\tau$ and directly implies that 
$\alpha_2 I_{n(n+p)} \preceq \sum_{\tau=t}^{t+T} \psi_\tau^\T \psi_\tau$, for some $\alpha_2>0$, as well. 
As a result $\sum_{\tau=t}^{t+T} \Xi_\tau^\T \Xi_\tau \preceq \alpha_3 T I_{n(n+p)}$
for some $0<\alpha_3<1$, and, therefore, 
$\sum_{\tau=t}^{t+T} \pbra{ \Xi_\tau^\T \Xi_\tau - I_{n(n+p)} } \preceq - \alpha_4 T I_{n(n+p)}$ for some $0<\alpha_4<1$.
Finally this implies that 
$\sbra{ \sum_{\tau=t}^{t+T} \Phi(\tau;t)^\T \pbra{ \Xi_\tau^\T \Xi_\tau - I_{n(n+p)}} \Phi(\tau;t) } \leq - \alpha_1 I_{n(n+p)}$
for some $0<\alpha_1<1$ \cite{anderson1969new}.
Notice that the second inequality of \eqref{eq:PE} is necessary to ensure non-singularity of the transition matrix 
$\Phi(\tau;t)$ \cite{jenkins2018convergence}.
Finally, as an immediate result of \eqref{eq:lyapunov-theta}, 
$V(t+T+1,\bar \theta_{t+T}+1) \leq (1-\alpha_1) V(t,\bar \theta_{t})$, $\forall t\geq 0$, which 
implies uniform asymptotic stability in the large, and, due to linearity, exponential stability in the large.
%


\vspace{-2em}
\begin{IEEEbiography}[{\includegraphics[width=1in,height=1.25in,clip,keepaspectratio]
{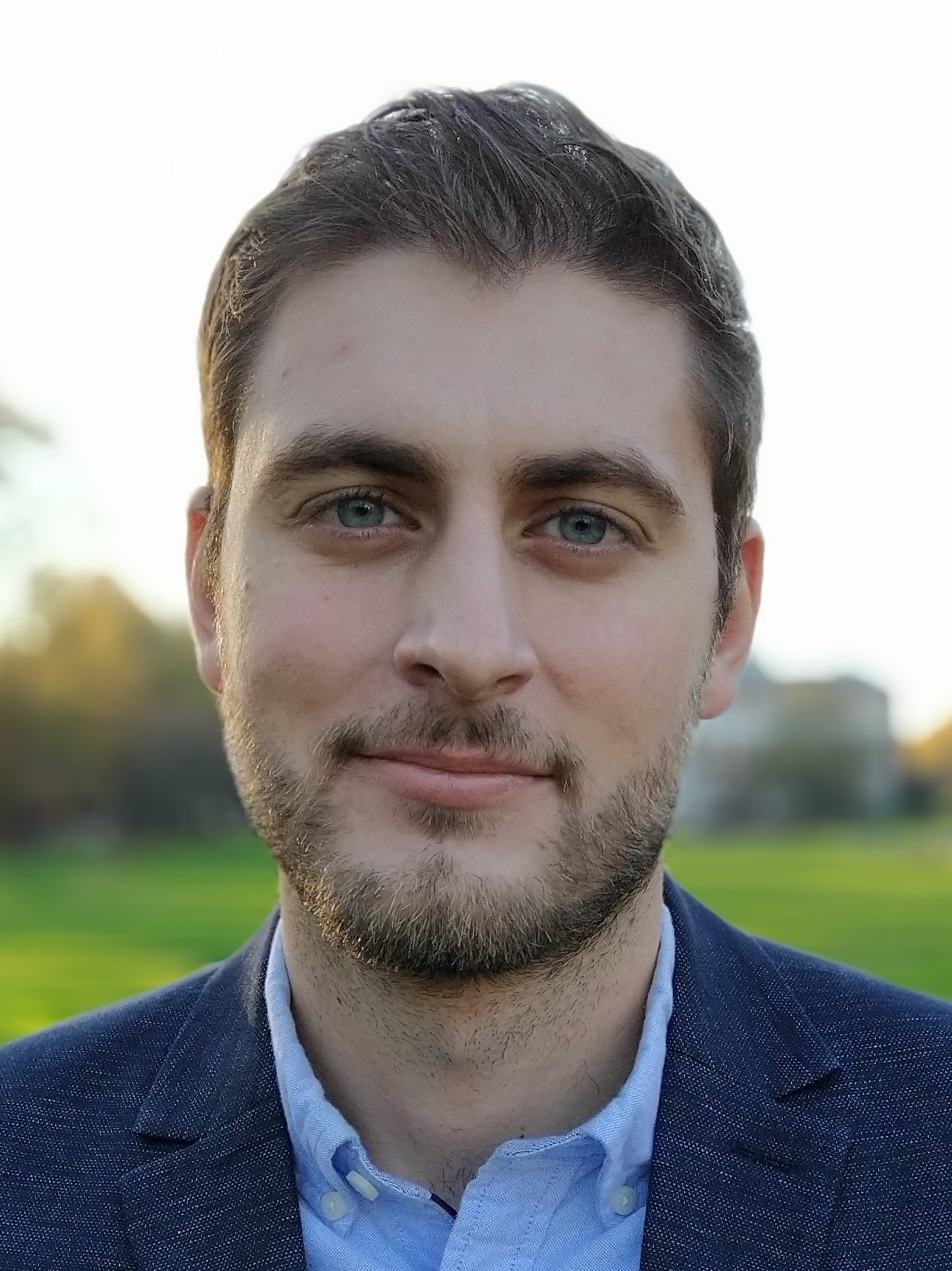}}]{Christos N. Mavridis}
received his Diploma in electrical and computer engineering from the National Technical University of Athens, Greece, in 2017,
and the M.S. and  Ph.D. degrees in electrical and computer engineering at the University of Maryland, College Park, MD, in 2021. 
His research interests include stochastic optimization, learning theory,
hybrid systems, and control theory,.

He is currently a postdoc at KTH Royal Institute of Technology, Stockholm, and has been affiliated as a research scientist with the Institute for Systems Research (ISR), University of Maryland, MD, the Nokia Bell Labs, NJ, the Xerox Palo Alto Research Center (PARC), CA, and Ericsson AB, Stockholm. 

Dr. Mavridis is an IEEE member, and a member of IEEE/CSS Technical Committee on Security and Privacy. He has received the A. James Clark School of Engineering Distinguished Graduate Fellowship and the Ann G. Wylie Dissertation Fellowship in 2017 and 2021, respectively. He has been a finalist in the Qualcomm Innovation Fellowship US, San Diego, CA, 2018, and he has received the Best Student Paper Award in the IEEE International Conference on Intelligent Transportation Systems (ITSC), 2021.
\end{IEEEbiography}

\vspace{-2em}
\begin{IEEEbiography}[{\includegraphics[width=1in,height=1.25in,clip,keepaspectratio]
{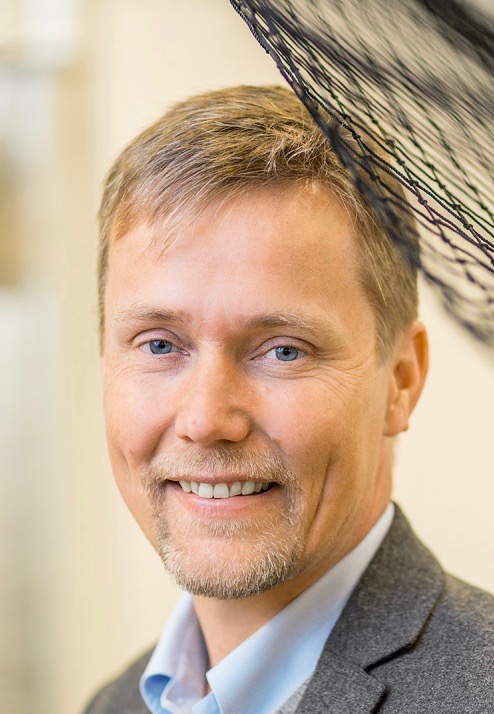}}]{Karl H. Johansson} 
 is Swedish Research Council Distinguished Professor
  in Electrical Engineering and Computer Science at KTH Royal
  Institute of Technology in Sweden and Founding Director of Digital
  Futures. He earned his MSc degree in Electrical Engineering and PhD in
  Automatic Control from Lund University. He has held visiting
  positions at UC Berkeley, Caltech, NTU and other prestigious
  institutions. His research interests focus on networked control systems and
  cyber-physical systems with applications in transportation, energy,
  and automation networks. For his scientific contributions, he has
  received numerous best paper awards and various distinctions from
  IEEE, IFAC, and other organizations. He has been awarded Distinguished
  Professor by the Swedish Research Council, Wallenberg Scholar by the
  Knut and Alice Wallenberg Foundation, Future Research Leader by the
  Swedish Foundation for Strategic Research. He has also received the
  triennial IFAC Young Author Prize and IEEE CSS Distinguished
  Lecturer. He is the recipient of the 2024 IEEE CSS Hendrik W. Bode
  Lecture Prize. His extensive service to the academic community includes being
  President of the European Control Association, IEEE CSS Vice
  President Diversity, Outreach \& Development, and Member of IEEE CSS
  Board of Governors and IFAC Council. He has served on the
  editorial boards of Automatica, IEEE TAC, IEEE TCNS and many other
  journals. He has also been a member of the Swedish Scientific Council for
  Natural Sciences and Engineering Sciences. He is
  Fellow of both the IEEE and the Royal Swedish Academy of Engineering
  Sciences.
\end{IEEEbiography}

\end{document}